\documentclass[journal,12]{IEEEtran}
\usepackage[pdftex]{graphicx}
\graphicspath{{../pdf/}{../jpeg/}}
\DeclareGraphicsExtensions{.pdf,.jpeg,.png}
\usepackage[cmex10]{amsmath}
\usepackage{amsthm}
\usepackage{graphicx,color}
\usepackage[dvipsnames]{xcolor}
\usepackage{graphicx}
\usepackage{mathtools}
\usepackage{amssymb,color}
\allowdisplaybreaks
\usepackage{subcaption}
\usepackage{cite}
\usepackage{blindtext}
\usepackage{tcolorbox}
\usepackage{colortbl}
\newcounter{mycounter}

\definecolor{Gray}{gray}{0.9}
\usepackage{xcolor}
\usepackage{lipsum}
\usepackage[font={small}]{caption}
\usepackage{algorithm}
\usepackage{algorithmic}
\usepackage{array} 
\usepackage{makecell}
\usepackage{tabularx}
\usepackage{bm}
\usepackage{siunitx, soul}

\usepackage{pgfplots}
\pgfplotsset{compat=newest}
\usetikzlibrary{plotmarks}
\usetikzlibrary{arrows.meta}
\usepgfplotslibrary {patchplots}
\usepackage{grffile}

\newtheorem{theorem}{Theorem}

\newtheorem{rem}{Remark}
\usepackage{booktabs}

\usepackage[font=scriptsize,labelfont=bf]{caption}
\usepackage{pifont}

\newcommand{\tr}{\text{Tr}}

\newcommand{\bc}{\text{BackCom}\xspace}

\newcommand{\abc}{\text{AmBC}\xspace} 
\newcommand{\thr}{\text{th}}

\makeatother
\bstctlcite{IEEEexample:BSTcontrol}
\newcommand{\qh}{\mathbf{h}}

\newcommand{\qf}{\mathbf{f}}
\newcommand{\qs}{\mathbf{s}}
\newcommand{\qg}{\mathbf{g}}
\newcommand{\qr}{\mathbf{r}}
\newcommand{\qu}{\mathbf{u}}

\newcommand{\qw}{\mathbf{w}}
\newcommand{\qx}{\mathbf{x}}

\newcommand{\qS}{\mathbf{S}}

\newcommand{\qG}{\mathbf{G}}
\newcommand{\qW}{\mathbf{W}}
\newcommand{\qQ}{\mathbf{Q}}
\newcommand{\qI}{\mathbf{I}}
\newcommand{\qR}{\mathbf{R}}
\newcommand{\qP}{\mathbf{P}}
\newcommand{\qalpha}{\boldsymbol{\alpha}}

\newcommand{\qH}{\mathbf{H}}

\newcommand{\qU}{\mathbf{U}}

\definecolor{LightGray}{gray}{0.9}
\definecolor{MediumGray}{gray}{0.5}
\definecolor{DarkGray}{gray}{0.2}

\usepackage{enumitem}
\usepackage{multirow}

\makeatother

\begin{document}

\title{Transmit Power Optimization for  Integrated Sensing and Backscatter Communication} 

\author{Shayan Zargari, Diluka Galappaththige, \IEEEmembership{Member, IEEE},  and  Chintha Tellambura, \IEEEmembership{Fellow, IEEE} 
\thanks{S. Zargari, D. Galappaththige, and C. Tellambura with the Department of Electrical and Computer Engineering, University of Alberta, Edmonton, AB, T6G 1H9, Canada (e-mail: \{zargari, diluka.lg, ct4\}@ualberta.ca). }  \vspace{-7mm}}

\maketitle

\begin{abstract} 
Ambient Internet of Things networks use low-cost, low-power backscatter tags in various industry applications. By exploiting those tags, we introduce the integrated sensing and backscatter communication (ISABC) system, featuring multiple backscatter tags, a user (reader), and a full-duplex base station (BS) that integrates sensing and (backscatter) communications. The BS undertakes dual roles of detecting backscatter tags and communicating with the user, leveraging the same temporal and frequency resources. The tag-reflected BS signals offer data to the user and enable the BS to sense the environment simultaneously.  We derive both user and tag communication rates and the sensing rate of the BS. We jointly optimize the transmit/received beamformers and tag reflection coefficients to minimize the total BS power. To solve this problem, we employ the alternating optimization technique. We offer a closed-form solution for the received beamformers while utilizing semi-definite relaxation and slack-optimization for transmit beamformers and power reflection coefficients, respectively. For example, with ten transmit/reception antennas at the BS, ISABC delivers a \qty{75}{\percent} sum communication and sensing rates gain over a traditional backscatter while requiring a \qty{3.4}{\percent} increase in transmit power. Furthermore, ISABC with active tags only requires a \qty{0.24}{\percent} increase in transmit power over conventional integrated sensing and communication.
\end{abstract}


\IEEEpeerreviewmaketitle
\section{Introduction}
Future Internet-of-Things (IoT) networks demand low-power, high-quality wireless connectivity, and precise, robust sensing capabilities \cite{Huawei_ambient, Huawei}. Ambient power-enabled (battery-free) IoT, a vibrant research area, has garnered attention, with 3GPP launching a dedicated study item \cite{Huawei_ambient, Huawei}. These networks link devices capable of autonomously sensing, collecting, and sharing environmental data, fostering real-world applications in smart homes, cities, autonomous vehicles, industrial IoT, healthcare, etc. \cite{HoangBook2020, Diluka2022, Rezaei2023Coding}. Such applications necessitate not only low-power communication but also advanced sensing functionalities. For instance, a smart-home temperature sensor can move and sense temperature in different locations, enabling the network to extract vital environmental information, such as range, velocity, or angle, for environment learning and mapping \cite{Diluka2023, Diluka2022}.

While ambient-IoT devices exhibit low power consumption and limited processing capabilities, they can help address the central challenge of enabling simultaneous communication and sensing. In this context, a groundbreaking concept introduced in \cite{Diluka2023} is known as \textit{Integrated Sensing and Backscatter Communications (ISABC)}. This paradigm shift merges the principles of integrated sensing and communication (ISAC) with the capabilities of backscatter communication (\bc), offering a solution that facilitates the concurrent execution of sensing and communication tasks in ambient power-enabled IoT networks.

Before proceeding to ISABC, we first briefly describe  ISAC and \bc.

\subsection{Integrated Sensing and Communication}

ISAC, or Integrated Sensing and Communication, represents a revolutionary shift from traditional network models, allowing concurrent sensing and communication tasks. Unlike conventional systems with segregated networks, ISAC seamlessly integrates these functions, which carry substantial implications for the transition beyond 5G into the domain of 6G \cite{Zhang2022, Wang2022}. This paradigm empowers devices to extract environmental insights from RF signals and reflections and facilitates innovative services like precise localization, activity tracking, object detection, urban traffic monitoring, and weather observations \cite{Zhang2022, Wang2022}.

Moreover, the environmental data acquired through ISAC enhances communication performance, enabling precise beamforming and rapid beam failure recovery. This sensing capability propels future IoT networks into the realm of perceptive networks, laying the groundwork for intelligence within the ISAC network and unlocking possibilities across various domains, including smart homes, cities, warehousing, healthcare, and beyond \cite{Zhang2022, Wang2022}.

The two categories of ISAC are  (i) Device-free ISAC and (ii) Device-based ISAC \cite{Liu2022ISAC}.
\begin{enumerate}
    \item \textit{Device-free ISAC}: This approach detects the sensing information of unregistered external targets (vehicles, animals, people, etc.). Unlike registered targets, these entities cannot transmit and/or receive sensing signals, making the sensing procedure independent of their transmission and/or reception capabilities.

    \item \textit{Device-based ISAC}: 
    Network-registered devices, including mobile phones, sensors, UAVs, etc., enable sensing functionality. The targets involved in sensing can transmit and/or receive signals, and the procedure relies on their transmission and/or reception. An illustrative case is wireless-based localization for locating mobile devices.
\end{enumerate}

Additionally, sensing in these ISAC categories can be subdivided based on the transmitter and sensing receiver configurations, i.e., mono-static, bi-static, and multi-static, as well as the type of sensing signal, i.e., active if the sensing receiver uses the reflected/diffracted signals of its own transmission or passive if it uses the received sensing signals from another transmitter \cite{Zhang2022, Wang2022}. In active sensing, the sensing receiver operates simultaneously while transmitting, i.e., in full-duplex (FD) mode \cite{Zhenyao2023}. In particular, self-interference (SI), a critical issue in FD operation, has a significant impact on sensing performance. For instance, in an FD ISAC system, the SI cancellation must be performed only for the direct signal coupling between the transceiver antennas, while preserving the target reflections \cite{Zhenyao2023}. Many developing SI cancellation approaches, including antenna isolation, analog cancellation, digital cancellation, and machine learning-based SI mitigation, can successfully suppress the SI \cite{Zhenyao2023, Mohammadi2023}.
 
\subsection{Backscatter Communication}
This technology,  especially for ambient-powered IoT networks, has drawn significant interest from both academic and industrial research communities \cite{HoangBook2020, Diluka2022, Rezaei2023Coding, Zargari9780612, Hakimi9877898}. It relies on passive tags devoid of active RF components, which communicate by reflecting external RF signals. This approach overcomes the limitations of battery-powered IoT devices, which often require frequent replacement or recharging, resulting in substantial maintenance costs, environmental concerns, and, in specific cases, safety hazards (e.g., wireless sensors in industries like power and petroleum). As a solution, the adoption of batteryless backscatter devices, such as passive tags, or devices with limited energy storage (semi-passive tags), holds promise for meeting the connectivity demands of future IoT networks and applications \cite{HoangBook2020, Diluka2022, Rezaei2023Coding}.

Due to using RF signals can be generated by dedicated or ambient sources, tags can be cost-effective, ultra-low-power devices (e.g.,  few \qty{}{\nW} to \qty{}{\uW}) \cite{HoangBook2020}. Moreover, \bc optimizes spectrum usage without the need for additional frequency spectrum allocation. However, the performance of ambient \bc  (\abc) can be hindered by interference from legacy signals. Nonetheless, employing a cooperative receiver/user capable of decoding both primary and backscatter data mitigates primary interference in \bc  \cite{Liang2022}.

\subsection{Integrated Sensing and Backscatter Communications}
While ISABC falls under the umbrella of device-based ISAC, it distinguishes itself from standard ISAC systems by substituting the sensing/radar target with backscatter tags to facilitate opportunistic sensing \cite{Diluka2023}. Table \ref{tab:comparison} provides a breakdown of the distinctions between these two. While ISAC may involve targets that neither transmit nor receive sensing signals, such as vehicles or birds, or devices that do, like mobile phones for wireless-based localization, ISABC stands out by exclusively utilizing backscatter tags. These tags provide environmental insights to the base station (BS) and furnish supplementary data to the user. To underscore the nuances, we note these key distinctions:
\begin{itemize}
\item ISABC uses the backscatter tag as a sensing instrument and a data provider. It can act as a sensor (e.g., monitoring temperature or humidity), conveying ambient information to the user. Simultaneously, the BS leverages the same tag signal to glean critical environmental metrics like range or velocity.

\item ISABC's hallmark lies in its capability to merge sensing and backscatter data. This convergence augments communication and sensing prowess and heightens computational demands due to the necessity for advanced decoding algorithms, especially those relying on successive interference cancellation (SIC).
\end{itemize}

In applications like smart homes, while tags can delineate the environment, the key idea is to exploit tag-reflected signals at the BS for enhanced sensing. This is achieved without incurring additional RF resources, escalating hardware expenses, or modifying tags.

\begin{table}[t]
    \centering
    \renewcommand{\arraystretch}{1.2} 
    \setlength{\tabcolsep}{10pt} 
    \captionsetup{font=small,labelfont={color=DarkGray,bf},textfont={color=DarkGray}}
    \caption{A comparison between ISAC and ISABC.}
    \begin{tabular}{lcc}
    \hline
    \rowcolor{LightGray}
    \textbf{Features} & \textbf{ISAC} & \textbf{ISABC} \\
    \hline
    Target & $\checkmark$ & $\times$ \\
    Tag & $\times$ &  $\checkmark$ \\
    Additional data at the user & $\times$ & $\checkmark$ \\
    Power allocation at the BS & $\checkmark$ & $\checkmark$ \\
    User decoding & Conventional & SIC \\
    Sensing signal & Active/Passive & Active \\
    \hline
    \end{tabular}
    \vspace{-0mm}
    \label{tab:comparison}
\end{table}

\subsection{Motivation and Our Contribution}
While many works study   ISAC and \bc systems separately  \cite{HoangBook2020, Diluka2022, Liu2022, Rahman2020, Rezaei2023Coding, Galappaththige2023SR, Zhitong2022}, study \cite{Diluka2023} is the first one to introduce a holistic exploration of their integrated functionalities and ensuing performance metrics of a limited ISABC system.  This work thus breaks new ground by exploiting their synergistic potential.

In \cite{Diluka2023}, the integration of sensing at an FD BS with a backscatter tag and user is detailed. The tag reflects the BS signal to transmit data to the user, while the BS extracts environmental data from the tag's signal. The study provides closed-form expressions for user and tag communication rates and BS sensing rates. However, \cite{Diluka2023} has not addressed multiple tags, BS transmit/received beamformers, and tag's reflection coefficients beyond the single-tag scenario. This study extends the research, exploring multi-tag energy harvesting (EH) scenarios and optimizing BS beamforming, received beamformers, and tag reflection coefficients.  Our contributions are summarized as follows:

\begin{enumerate}
    \item  The objective of this paper is to intertwine sensing functions with communication capabilities and to elucidate the advantages of  ISABC. To this end, we analyze a network of multi-tags, a user, and an FD BS. The BS manages communication for both the user and the tags. Specifically, the tags reflect the BS’s signal to communicate with the user, while the BS exploits the same reflected signal to derive environmental insights.

    \item  We optimize the system for minimal BS power consumption while meeting each node's quality-of-service (QoS) requirements. The optimization variables are the  BS transmit/received beamformers and tag power reflection coefficients. The nonlinear EH model at the tags further complicates this non-convex optimization problem.
    To tackle this, we use an AO (alternative optimization)  approach \cite{bezdek2003convergence}. We start by optimizing the BS received beamformer for the tags' signal using minimum mean-squared error (MMSE) filtering and the generalized Rayleigh quotient form of the signal-to-noise-to-interference ratio (SINR) \cite{Stanczak2008book, Wan2016}. Then, we compute the BS transmit beamformers using the semidefinite relaxation (SDR) method \cite{so2007approximating, Qingqing}.  Finally, we introduce a slack-optimization problem to optimize the tag reflection coefficients \cite{Shayan_Zargari, Razaviyayn2013}. 

    \item We provide convergence and complexity analysis and simulations to assess the efficiency of ISABC, comparing it to conventional ISAC, communication-only, and sensing-only schemes (with/without EH). With a configuration of ten antennas each for transmission and reception at the BS, ISABC offers a $75\%$ combined communication and sensing rate enhancement compared to the conventional \bc, while only necessitating a modest $3.4\%$ rise in transmit power.

\end{enumerate}

\textit{Notation}: 
Vectors and matrices are expressed by
boldface lower case letters $\mathbf{a}$ and capital letters $\mathbf{A}$, respectively. For a square matrix $\mathbf{A}$, $\mathbf{A}^{\rm{H}}$ and $\mathbf{A}^{\rm{T}}$ are Hermitian conjugate transpose and transpose of a matrix, respectively. $\mathbf{I}_M$ denotes the $M$-by-$M$ identity matrix. The Euclidean norm of a complex vector and the absolute value of a complex scalar are denoted by $\|\cdot\|$ and $|\cdot|$, respectively. The distribution of a circularly symmetric complex Gaussian (CSCG) random vector with mean $\boldsymbol{\mu}$ and covariance matrix $\mathbf{C}$ is denoted by $\sim \mathcal{C}\mathcal{N}(\boldsymbol{\mu},\,\mathbf{C})$. 
The expectation operator is denoted by $\mathbb{E}[\cdot]$. Besides, $\mathbb{C}^{M\times N}$ and ${\mathbb{R}^{M \times 1}}$ represent $M\times N$ dimensional complex matrices and $M\times 1$ dimensional real vectors, respectively. Further, $\mathcal{O}$ expresses the big-O notation. Finally, $\mathcal{K} \triangleq \{1,\ldots,K\}$ and $\mathcal{K}_k \triangleq \mathcal{K}\setminus\{k\}$.

\begin{figure}[!t]
    \centering 
    \def\svgwidth{220pt} 
    \fontsize{9}{9}\selectfont 
    \graphicspath{{}}
    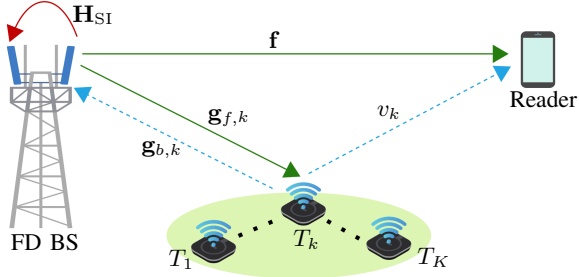  
    \caption{An ISABC system setup. \vspace{-0mm}}  \label{fig_SystemModelMultiTag}
\end{figure}

\section{System, Channel, and Signal Models}\label{Sec_system_modelA}
Here, we describe the system, channel, and transmission models in detail.

\subsection{System Model}

As shown in Fig.~\ref{fig_SystemModelMultiTag}, we consider an ISABC network having an FD BS consisting of $M\geq 1$ transmit and $N\geq 1$ receiver uniform linear array (ULA) antennas, $K$ single-antenna backscatter tags/sensors, denoted by $T_k$, $\forall k \in \mathcal{K} \triangleq \{1,\ldots, K\}$, and a single-antenna user (or mobile reader). The BS antennas are spaced at half-wavelengths \cite{Zhenyao2023}. Tags perform EH and backscatter data to the reader (Section~\ref{sec_data_EH_tag}).
 
The FD BS uses transmit beamforming for communication and environment sensing. Tags utilize harvest energy from the BS signal and also reflect it for data transfer. A cooperative user decodes its data and then uses SIC for tag data. The BS also captures the tag-reflected signals to extract environmental insights \cite{Diluka2023, Liu2022, Rahman2020}. It has separate antennas for transmission and reception to limit SI, assuming perfect cancellation and synchronized timing \cite{Liu2022, Rahman2020}.

\subsection{Channel Model}
We consider block flat-fading channel models for the system. During each fading block, the channels between the BS and the user, the BS and $T_k$, and $T_k$ and the user are denoted by $\mathbf{f} \in \mathbb{C}^{M\times 1}$, $\mathbf{g}_{f,k} \in \mathbb{C}^{M\times 1}$, and $v_k \in \mathbb{C}$, respectively. Moreover, $\mathbf{g}_{b,k} \in \mathbb{C}^{N\times 1}$ represents the channel between $T_k$ and the BS receiver antennas. Among these channels, pure communication channels, i.e., $\mathbf{f}$ and $v_k$, are modeled as Rayleigh fading and given by 
\begin{eqnarray}\label{eqn_channel_model}
    \mathbf{a} = \zeta_a^{1/2} \tilde{\mathbf{a}},
\end{eqnarray}
where $\mathbf{a} \in \{\mathbf{f}, v_{k}\}$. In \eqref{eqn_channel_model}, $\zeta_a$ captures the large-scale path-loss and shadowing,  which stays constant for several coherence intervals. Moreover, $\tilde{\mathbf{a}} \sim \mathcal{CN}(\mathbf{0}, \mathbf{I}_{A})$ accounts for the small-scale Rayleigh fading\footnote{Note that $v_{k} = \zeta_{v_{k}}^{1/2} \tilde{v}_{k}$ and $\tilde{v}_{k} \sim  \mathcal{CN}\left(0,1 \right)$.}, where $A\in \{M,1\}$.

On the other hand, following the echo signal multiple-input and multiple-output (MIMO) radar, the channels between the BS and tags are modeled as line-of-sight (LoS) paths \cite{Zhenyao2023}. We denote the transmit/receiver array steering vectors to the direction $\theta_k$ by
\begin{eqnarray}\label{array_resp}
    \mathbf{b}(\theta_k)= \sqrt{\frac{\zeta_b}{B}} \left[1, e^{j\pi \sin(\theta_k)}, \ldots, e^{j\pi (B-1) \sin(\theta_k)} \right]^{\rm{T}},
\end{eqnarray}
where $\mathbf{b} \in \{\mathbf{g}_{f,k}, \mathbf{g}_{b,k}\}$, $B\in\{M, N\}$, $\theta_k$ is the direction of $T_k$ with respect to the BS and user/reader direction, and $\zeta_b$ is the path-loss. Finally, the SI channel between the transmitter and the receiver ULAs of the  BS is denoted as $\mathbf{H}_{\rm{SI}} \in  \mathbb{C}^{M\times N}$.

{In typical {\bc} applications, tags are often deployed in stable environments such as rooms or warehouse shelves, maintaining a clear LoS to the BS. This setup supports applications like inventory tracking and smart shelving. Conversely, communication links, including BS-user and tags-user connections, encounter various propagation challenges due to mobility, obstacles, and varying distances. Hence, the Rayleigh fading model is commonly employed due to its stochastic nature, effectively capturing the diverse signal paths encountered in urban landscapes \mbox{\cite{Diluka2023, Zhenyao2023, Shuo2018, rezaei2023timespread, Zargari10353962}}. However, it is worth noting that the proposed optimization framework and its solution are adaptable to any fading model. Additionally, performance trends remain relatively consistent across different fading models (Fig.~{\ref{Rician_fig}}).}

\begin{rem}
{We assume that channel estimation and data transmission tasks occur in two separate time slots. In the initial slot, channel state information (CSI) can be estimated using emerging techniques \mbox{\cite{rezaei2023timespread, Abdallah2023, Shuo2018}}. These methods encompass pilot-based, blind, and semi-blind approaches, employing algorithms like least squares (LS), MMSE estimator, expectation maximization (EM), and eigenvalue decomposition (EVD) to achieve high-precision channel estimation \mbox{\cite{rezaei2023timespread, Abdallah2023, Shuo2018}}. However, as the focus of our study lies in the integration of sensing into \bc, we assume the presence of perfect CSI, implying knowledge of $\qf$ and ${\qh_k}_{k\in\mathcal{K}}$. This assumption aligns with standard practice in most studies \mbox{\cite{Azar_Hakimi}}.}  
\end{rem}

\subsection{Tag Characteristics} \label{sec_data_EH_tag}
Each tag employs load modulation, dependent on the complex reflection coefficient of $T_k$, $\forall k \in \mathcal{K}$ \cite{Diluka2022, Rezaei2023Coding}: 
\begin{equation}\label{tag:ref}
{\Theta_{i,k} = \frac{Z_i - Z_{a,k}^\star}{Z_i + Z_{a,k}}.}
\end{equation}
In \eqref{tag:ref}, $Z_{a,k}$ symbolizes the antenna impedance of $T_k$, and $Z_i$ represents the $i$-th load impedance where {$i \in \{1,2,\ldots, Q\}$}.  We can express  {$\Theta_{i,k} = |\Theta_{k}| e^{j \varphi_{i}}$, where $\varphi_i \in [0, 2\pi]$}. Since all the tags utilize a uniform set of phases {$\{\varphi_1, \varphi_2, \ldots, \varphi_Q\}$} to transmit their data, {$Q$-ary phase-shift keying (PSK)} constellation is realized. This means the phase {$\angle \Theta_{i,k}$} is contingent not on a specific tag but rather on {$ \varphi_i$}.

Furthermore, the index $i$ is omitted in the power reflection coefficient {$|\Theta_{i,k}|^2$} as it is feasible to design a collection of load impedance values that maintain a consistent {$|\Theta_{i,k}|^2$}, only modulating the phase of {$\Theta_{i,k}$} \cite{Azar_Hakimi}. Thus, we denote  {$\alpha_k = |\Theta_{i,k}|^2$}, which is the power reflection coefficient of $T_k$.  The tag design will satisfy $0 < \alpha_k < 1$ \cite{Rezaei2023Coding}.

\subsection{Transmission Model}
The BS transmitted signal $\mathbf{x} \in \mathbb{C}^{M\times 1}$, which includes both data and sensing waveforms, is given by $\mathbf{x} = \mathbf{w}x_d + \mathbf{s}$, where $x_d \in \mathbb{C}$ is the intended data symbol for the mobile user/reader with unit power, i.e., $ \mathbb{E}\{\vert x_d \vert^2 \}=1$, $\mathbf{w} \in \mathbb{C}^{M\times 1}$ is the BS beamforming vector, and $\mathbf{s}\in \mathbb{C}^{M\times 1}$ is the sensing signal with the covariance matrix $\mathbf{S} \triangleq \mathbb{E}\{\mathbf{s} \mathbf{s}^{\rm{H}} \}$ for extending the degrees-of-freedom of $\mathbf{x}$ to achieve enhanced sensing performance \cite{Zhenyao2023, Haocheng2023}. 
Also, it is assumed that $x_d$ and $\mathbf{s}$ are independent of each other, and the beamforming at the BS is achieved through designing $\mathbf{w}$ and $\mathbf{S}$ \cite{Zhenyao2023, Haocheng2023}. The designed $\mathbf{S}$ can be used to generate the dedicated sensing signal, $\mathbf{s}$ \cite{Stoica2007}.

The user receives the BS signal as well as the tags' backscattered signals. The propagation delay differences for all signals are assumed to be negligible \cite{Liao2020}. The user-received signal is thus given by
\begin{eqnarray}\label{eqn_rx_user}
    y = \mathbf{f}^{\rm{H}} \mathbf{x} + \sum_{k\in \mathcal{K}} \sqrt{\alpha_k} \mathbf{h}_k^{\rm{H}} \mathbf{x} c_k + z_u,
\end{eqnarray}
where the first and the second terms in \eqref{eqn_rx_user} represent the direct-link, i.e., BS-to-user, and backscatter-link, i.e., BS-to-tags-to-user, signals, respectively. Moreover, $z_u \sim \mathcal{CN}(0,\sigma^2)$ is the white Gaussian noise (AWGN) at the user with $0$ mean and $\sigma^2$ variance, $\mathbf{h}_k$ is the effective backscatter channel through $T_k$, i.e., $\mathbf{h}_k= \mathbf{g}_{f,k}(\theta_k) v_k$, and $c_k$ is $T_k$'s data with $\mathbb{E}\{\vert c_k \vert^2 \}=1$. As the mobile user/reader and BS are integral to the primary networks, a pre-existing connection facilitates the exchange of information, including sensing waveform, via a control link \cite{positioningLTE}. It is thus assumed that the user knows the sensing waveform in advance and removes it before decoding data. Following the removal of the sensing signal, the received signal can be expressed as
\begin{eqnarray}\label{eqn_rx_noSens}
    y = \mathbf{f}^{\rm{H}} \mathbf{w} x_d  + \sum_{k\in \mathcal{K}} \sqrt{\alpha_k} \mathbf{h}_k^{\rm{H}} \left(\mathbf{w}x_d +\mathbf{s} \right) c_k + z_u.
\end{eqnarray}
Next, the user performs SIC to recover the backscattered data from the tags. In particular, the user decodes its own signal, treating tag signals as interference, and then subtracts the decoded $x_d$ from the received signal \eqref{eqn_rx_noSens} for decoding the tags' data. The post-processed signal for decoding tags' data is thus given as
\begin{eqnarray}\label{eqn_rx_tag}
    y_t =  \sum_{k\in \mathcal{K}} \sqrt{\alpha_k} \mathbf{h}_k^{\rm{H}} \left(\mathbf{w}x_d +\mathbf{s} \right) c_k + z_u.
\end{eqnarray}
Backscattered signals from tags, on the other hand, reach not only the user but also the BS. The BS thus aims to extract environmental information from these unintentionally received backscattered signals \cite{Diluka2023}. The received signal at the BS, i.e., $\mathbf{y}_b \in \mathbb{C}^{N\times 1}$, is given as
\begin{eqnarray}\label{eqn_rx_BS}
    \mathbf{y}_b = \sum_{k\in \mathcal{K}} \sqrt{\alpha_k} \mathbf{G}_k (\theta_k) \mathbf{x} c_k + \mathbf{H}_{\rm{SI}}^{\rm{H}} \mathbf{x}  + \mathbf{z}_b,
\end{eqnarray}
where $\mathbf{G}_k (\theta_k) \triangleq \mathbf{g}_{b,k}(\theta_k) \mathbf{g}_{f,k}^{\rm{H}}(\theta_k)$ and  $\sqrt{\alpha_k} \mathbf{G}_k (\theta_k) \mathbf{x} c_k$ is the $k$-th backscatter tag ($T_k$) reflection. The second term in \eqref{eqn_rx_BS} denotes the SI at the receiver of the BS due to simultaneous transmission and reception, and $\mathbf{z}_b \sim \mathcal{CN}(\mathbf{0},\sigma^2 \mathbf{I}_N)$ the AWGN at the BS. We assume that the FD BS cancels the SI at its receiver using perfect SI cancellation techniques \cite{Liu2022, Rahman2020}. The post-processed SI cancelled signal is thus given as
\begin{eqnarray}\label{eqn_rx_BS_SI}
    \mathbf{y}_b' = \sum_{k\in \mathcal{K}} \sqrt{\alpha_k} \mathbf{G}_k (\theta_k) \mathbf{x} c_k + \mathbf{z}_b.
\end{eqnarray}
The BS then applies the receiver beamformer, $\mathbf{u}_k \in \mathbb{C}^{N\times 1}$ for $k \in \mathcal{K}$,  to the received signal \eqref{eqn_rx_BS_SI} to capture the desired reflected signal of $T_k$. The post-processed signal for obtaining $T_k$'s sensing information is given as
\begin{eqnarray}\label{t_kSens}
    {y}_{b,k} \!&=&\! \mathbf{u}_k^{\rm{H}} \mathbf{y}_b'  \\
    \!&=&\!\! \sqrt{\alpha_k} \mathbf{u}_k^{\rm{H}} \mathbf{G}_k (\theta_k) \mathbf{x} c_k \!+\! \sum_{i\in \mathcal{K}_k} \!\!\sqrt{\alpha_i} \mathbf{u}_k^{\rm{H}} \mathbf{G}_i (\theta_i) \mathbf{x} c_i \!+\! \mathbf{u}_k^{\rm{H}} \mathbf{z}_b, \nonumber 
\end{eqnarray}
where $\mathcal{K}_k \triangleq \mathcal{K}\setminus\{k\}$.

\begin{rem}
Before proceeding, we wish to clarify the following assumptions used in the considered system: (i) When the transmit and receive arrays are colocated, the angles of a backscatter tag/target seen at the BS transceiver are the same in \eqref{eqn_rx_BS}, which is a reasonable and common assumption \cite{Jian2007}, (ii) Because the BS and the user are linked via a controlled link, the user is aware of the sensing waveform in advance \cite{positioningLTE}, and (iii) We assume that a dedicated channel estimating phase is used prior to FD transmission, ensuring CSI is available for beamforming design and SI at the BS, as well as SIC at the user/reader \cite{Temiz2022}.
\end{rem}

{Existing channel estimating methods accurately estimate the cascaded channel $\qh_k$, but do not separate or estimate the individual channels, i.e., $\qg_{f,k}$ and $v_k$. Although $\qh_k$ contains $\theta_k$, accurately estimating $\theta_k$ from $\qh_k$ is not feasible as it contains an unknown channel $v_k$. Consequently, as we proposed in this study, estimating the tags' sensing parameters necessitates a separate sensing framework. Conversely, the properties of $T_k$'s reflected signal $\sqrt{\alpha_k} \mathbf{G}_k (\theta_k) \mathbf{x} c_k$, such as round trip delay time and angle of arrival, can be used to acquire tag environmental information, such as range, velocity, and angle. Maximizing the echo signal strength for a particular tag while limiting tag interference at the BS, i.e., sensing SINR, improves tag detection probability and precise estimate of these targeted parameters. However, as we primarily focus on integrating sensing into {\bc} systems and its performance optimization (transmit power), we leave the sensing parameter estimation for future research.}

\section{Communication and Sensing Performance}
The  SINRs for sensing and communication tasks substantially impact the performance of both systems. Herein, we derive those SINRs of the tags and the user to evaluate and optimize the ISABC system. 

\subsection{Communication Performance}
The main communication SINRs are the user SINR and tags' SINRs.
\subsubsection{User SINR}
The user first decodes its data, considering the tags' signal as interference. From \eqref{eqn_rx_noSens}, the received SINR is obtained as
\begin{eqnarray}\label{eqn_user_sinr}
    \Gamma_{u} = \frac{\vert \mathbf{f}^{\rm{H}} \mathbf{w}\vert^2}{\sum_{k\in \mathcal{K}} \alpha_k (\vert \mathbf{h}_k^{\rm{H}} \mathbf{w} \vert^2 +  \mathbf{h}_k^{\rm{H}} \mathbf{S}\mathbf{h}_k) + \sigma^2}.
\end{eqnarray}
Here, we consider that while the user can cancel interference from the direct-link sensing signal, i.e., $\mathbf{f}^{\rm{H}} \mathbf{s}$, it cannot cancel interference from tag-sensing reflections because they consist of unknown tag data, i.e., $c_k$ for $k\in \mathcal{K}$.

\subsubsection{$T_k$'s SINR}
{The user employs the SIC for decoding the backscatter data. Using {\eqref{eqn_rx_tag}}, the SINR of $T_k$ at the user is given as}
\begin{eqnarray}\label{eqn_tag_sinr}
    \Gamma_{t,k} = \frac{\alpha_k \left( \vert \mathbf{h}_k^{\rm{H}} \mathbf{w}  \vert^2  + \mathbf{h}_k^{\rm{H}} \mathbf{S} \mathbf{h}_k\right)  }{ \sum_{i \in \mathcal{K}_k} \alpha_i \left( \vert \mathbf{h}_i^{\rm{H}} \mathbf{w}  \vert^2  + \mathbf{h}_i^{\rm{H}} \mathbf{S} \mathbf{h}_i\right) + \sigma^2}.
\end{eqnarray}

\subsection{Sensing Performance}
The BS uses the unintentionally received backscattered signal for sensing, i.e., to learn and obtain environmental information. The BS applies a receiver beamformer, $\mathbf{u}_k$ for $k \in \mathcal{K}$,  to the received signal \eqref{eqn_rx_BS_SI} to capture the desired reflected signal of $T_k$. To this end, the sensing SINR of $T_k$ is obtained using  \eqref{t_kSens} which is given by
\begin{eqnarray}\label{eqn_tag_sen_SINR}
    \Upsilon_{k} &=& \frac{\alpha_k \mathbb{E} \left\{ \vert \mathbf{u}_k^{\rm{H}} \mathbf{G}_k (\theta_k) \mathbf{x} \vert^2 \right\}  }{ \sum\limits_{i\in \mathcal{K}_k} \alpha_i \mathbb{E} \left\{ \vert \mathbf{u}_k^{\rm{H}} \mathbf{G}_i(\theta_i) \mathbf{x} \vert^2 \right\} + \mathbb{E} \left\{ \vert \mathbf{u}_k^{\rm{H}} \mathbf{z}_b \vert^2 \right\}   } \nonumber \\
    &=& \frac{\alpha_k \mathbf{u}_k^{\rm{H}} \mathbf{G}_k (\theta_k) \mathbf{R}_x  \mathbf{G}_k^{\rm{H}} (\theta_k) \mathbf{u}_k  }{ \mathbf{u}_k^{\rm{H}} \left( \sum\limits_{i\in \mathcal{K}_k} \alpha_i  \mathbf{G}_i(\theta_i) \mathbf{R}_x \mathbf{G}_i^{\rm{H}} + \sigma^2 \mathbf{I}_N  \right)\mathbf{u}_k  },  \qquad
\end{eqnarray}
where $\mathbf{R}_x \triangleq \mathbb{E} \{\mathbf{x} \mathbf{x}^{\rm{H}} \} = \mathbf{w}\mathbf{w}^{\rm{H}} + \mathbf{S}$ is the covariance matrix of the BS transmitted signal \cite{Zhenyao2023}.

\begin{rem}
{Communication and sensing performance are essentially determined by the associated SINRs. In particular, communication symbol detection probability increases monotonically with SINR \mbox{\cite{Goldsmith_2005, Tse_Viswanath_2005}}. Maximizing SINR eventually minimizes the symbol error probability. Therefore, we use the communication SINR performance as a standard metric. Similarly, in sensing, the detection probability of a target (tag) is proportional to its sensing SINR \mbox{\cite{Zhenyao2023, Cui2014}}. The sensing SINR enables target detection using both transmit and receiver beamforming (see Fig.~{\ref{beam_pro}} and Fig.~{\ref{beam_sensing}}). It also aids in reducing interference between targets. However, the standard mean squared error of the transmit beampattern does not account for the receiver beampattern or target interference. Given the benefits of sensing SINR, we employ it as a viable metric for sensing performance.}
\end{rem}

\subsection{Tag's EH Model}
As mentioned before, the tags are passive and do not generate RF signals. Thus, they do not require batteries and rely entirely on EH to power their essential functions. They transmit their data by simply reflecting (i.e., backscattering) an external RF signal, which requires negligible power consumption.  The tags concurrently perform both EH and data communication operations via power-splitting of the incident RF signal \cite{Zhang2013, Azar_Hakimi}.

The power-splitting operation can be described as follows: let the incident RF  power at tag  $T_k$  be $p_{k}^{\rm{in}}=|\qg^{\rm{H}}_{f,k}\qw|^2+\qg^{\rm{H}}_{f,k}\qS\qg_{f,k}$. The tag reflects a fraction of $p_{k}^{\rm{in}}$ and harvests the remainder \cite{Zhang2013}. These amounts can be quantified as follows. 
\begin{enumerate}
    \item The reflected power is  $\alpha_k p_{k}^{\rm{in}}$, which is used  for data transmission, 
    \item The harvested power,  $p_{k}^{\rm{h}}$, can be modeled as a linear or nonlinear function of $p_{k}^{\rm{in}}$.  The linear model estimates the harvested power at each tag as   $p_{k}^{\rm{h}} = \eta(1-\alpha_k) p_{k}^{\rm{in}}$, where $\eta \in (0,1]$ is the power conversion efficiency. Although the linear model is the most widely used in the literature due to its simplicity, it ignores the nonlinear characteristics of actual EH circuits such as saturation and sensitivity \cite{Wang2020WPCN}. 
\end{enumerate}
Consequently, a parametric nonlinear sigmoid EH has been widely used  \cite{Boshkovska2015}. It models the total harvested power at $T_k$ as $p_{k}^{\rm{h}} = \Phi((1-\alpha_k)p_{k}^{\rm{in}})$, where
\begin{align}
\Phi(p_{k}^{\rm{in}}) &= \psi^{\text{NL}}_k - \frac{M_{\rm{NL}}\Omega_{\rm{NL}}}{1 - \Omega_{\rm{NL}}}, \quad \Omega_{\rm{NL}} = \frac{1}{1 + \exp(a_{\rm{NL}}b_{\rm{NL}})}, \label{eq:8a} \\
\psi^{\text{NL}}_k &= \frac{M_{\rm{NL}}}{1 + \exp(-a_{\rm{NL}}(p_{k}^{\rm{in}} - b_{\rm{NL}}))}, \quad \forall k, \label{eq:8b}
\end{align}
where $\psi^{\text{NL}}_k $ is a standard logistic function with constant $ \Omega_{\rm{NL}} $ ensuring a zero input/output response. Parameters $ a_{\rm{NL}} $ and $ b_{\rm{NL}} $ represent circuit characteristics like capacitance and resistance. $ M_{\rm{NL}} $ is the maximum harvested power when the EH circuit is saturated. Parameters $a_{\rm{NL}}$, $b_{\rm{NL}}$, and $M_{\rm{NL}}$ can be derived using a curve fitting tool \cite{Boshkovska2015}.  Other non-linear models may be found in \cite{Wang2020WPCN}. However, we must add that our problem formulation can handle linear and non-linear models within one unified framework.

Regardless of the choice of a linear or non-linear model, another critical parameter is the activation threshold, i.e., $p_b$. It is the minimal power required to wake up the EH circuit, which is typically \qty{-20}{\dB m} for commercial passive tags \cite{Diluka2022}. Thus, to activate the tag, the harvested power should exceed the threshold, i.e., $p_{k}^{\rm{h}} \ge p_b$. 
In particular, $(1-\alpha_k) p_{k}^{\rm{in}} \geq  p_b'$, where $p_b' \triangleq \Phi^{-1}(p_b)$ and $ \Phi^{-1}(p_b)= b_{\rm{NL}} - \frac{1}{a_{\rm{NL}}} \ln \left( \frac{M_{\rm{NL}} - p_b}{p_b} \right), \forall k$.
Without loss of generality, the nonlinear EH model is adopted for formulating the optimization problem and resource allocation algorithm design in the following.

\section{Problem Formulation}
Our objective is to optimize the BS received beamformers,  $  \{\qu_k\}_{k\in \mathcal{K}}$, alongside the transmit beamforming  $\qw$ and $ \qS$, and the tag  reflection coefficients, $\{\alpha_k\}_{k\in \mathcal{K}}$. We  denote the set of these optimization variables as $\mathcal{A} = \left\{ \{\qu_k\}_{k\in \mathcal{K}},  \{\alpha_k\}_{k\in \mathcal{K}}, \qw, \qS \succ 0 \right\}$. We focus on minimizing the total BS transmit power. This objective promotes large-scale connectivity, utilizing the saved power to enhance network capacity for additional tags and users. Such an approach is particularly beneficial for green IoT networks, aiming to reduce energy consumption, thus extending network lifespans, reducing costs, and enhancing resource efficiency \cite{Toro2022}.

This goal is achieved by ensuring that the communication SINR requirements for both tags and the user are met at the user, as well as the EH requirements of the tags and the sensing SINR requirements at the BS. The problem is thus formulated as follows:
\begin{subequations}
\begin{align}\label{P1}
\text{(P1)}:~& \min_{\mathcal{A}} \quad   \|\qw\|^2 + \tr(\qS),  \\
\text{s.t} \quad & \Upsilon_{k}  \geq \Upsilon_{k}^{\thr},~ \forall k, \label{P1:sinr_tag}  \\
& \Gamma_u \geq \Gamma_u^{\thr}, \label{P1:sinr_user} \\
& \Gamma_{t,k}  \geq  \Gamma_k^{\thr},~ \forall k, \label{P1:sinr_tag_user} \\
& p_{k}^{\rm{in}}  \geq  \frac{\Phi^{-1}(p_b)}{1-\alpha_k},~ \forall k,\label{P1:EH} \\
&\| \qu_k \|^2=1,~ \forall k, \label{P1:uk} \\
&0 < \alpha_k < 1,~ \forall k, \label{P1:alpha}
\end{align}
\end{subequations}
where \eqref{P1:sinr_tag}  guarantees the sensing SINR requirement of each tag in which $\Upsilon_{k}^{\thr}$ denotes the targeted sensing SINR of $T_k$ at the BS. 

{On the other hand, {\eqref{P1:sinr_user}} and {\eqref{P1:sinr_tag_user}} set the targeted SINR values, i.e., $\Gamma_u^{\thr}$ and $\Gamma_k^{\thr}$, for the user to decode its own data and tag data, respectively. These ensure the minimum quality of the rate for the user and tags.} 
Constraint \eqref{P1:EH}  indicates the minimum incident power required at each tag for activation.  Constraint \eqref{P1:alpha} specifies the natural bounds on the reflection coefficient of each tag.

\begin{rem}The justification for objective (15a) is as follows. Minimizing BS transmit power is crucial for energy conservation, cost reduction, and prolonged network lifespan. While other metrics like latency, throughput, and reliability are relevant, integrating them may introduce conflicting goals. Our tailored problem formulation prioritizes SINR requirements while minimizing BS transmit power, aligning with the efficiency and robustness essential for green IoT deployments.
\end{rem}

{Note that the transmit power optimization approaches of ISABC and ISAC differ and present unique challenges. Unlike ISAC with conventional targets, ISABC utilizes backscatter tags for sensing and as a data transmission medium. Hence, it adds complexity and additional constraints to its optimization problem, i.e., \mbox{\text{(P1)}}. In particular, ISABC has additional constraints \mbox{{\eqref{P1:sinr_tag_user}}} and \mbox{\eqref{P1:EH}} for tag data transmission and EH in comparison to conventional ISAC. In contrast, by omitting these constraints, ISAC offers a significantly simplified optimization framework. In addition, ISAC has a lower spectral efficiency than ISABC due to the absence of tag data transmission.}

Since the denominators of SINRs in \eqref{eqn_user_sinr} and \eqref{eqn_tag_sinr} are similar (except for the $k$-th term), constraints \eqref{P1:sinr_user} and \eqref{P1:sinr_tag_user} can be combined into a single constraint without changing the original problem (P1). Combining these constraints utilizing their similar structures thus yields the following equivalent  optimization problem:
\begin{subequations}
\begin{align}\label{P2}
\text{(P2)}:~& \min_{\mathcal{A}} \quad   \|\qw\|^2 + \tr(\qS),  \\
\text{s.t}  \quad & \frac{|\mathbf{f}^{\rm{H}} \mathbf{w} |^2}{\Gamma_u^{\thr}(1+\Gamma_k^{\thr})}-  \sum_{i \in \mathcal{K}_k} \alpha_i \left(|\qh_i^{\rm{H}} \qw|^2+\qh^{\rm{H}}_i \qS \qh_i \right)  \geq \sigma^2, \forall k  , \label{P2:sinr_user} \\
& \eqref{P1:sinr_tag},~ \eqref{P1:EH}-\eqref{P1:alpha},
\end{align}
\end{subequations}
where constraint \eqref{P1:sinr_user} and \eqref{P1:sinr_tag_user} are combine into constraint \eqref{P2:sinr_user}. Our next step is to develop a new optimization algorithm to solve \eqref{P2}.

\section{Proposed Solution}\label{pro_solu}
Problem \eqref{P2} is non-convex due to its constraints set, which involves products of the optimization variables. To tackle this, we turn to AO \cite{bezdek2003convergence}. It divides an optimization problem into sub-problems that are easier to solve individually, which are then solved alternatively, one at a time, while keeping the other variables fixed. The process continues iteratively until convergence or a stopping criterion is met. This approach works when a direct or simultaneous optimization of all variables is challenging or computationally expensive \cite{bezdek2003convergence}. Thus, to solve $\min_{x} f(x)$, where $x \in \mathbb{R}^{s}$ can be divided into $l > 1$ blocks, i.e., $x = (x_1, x_2, \ldots, x_l)^{\rm{T}}$ with $x_l \in \mathbb{R}^{s_k}$ and $\sum_{k=1}^l s_k = s$, the strategy is to cyclically minimize for one block at a time, holding the others constant, until convergence is achieved. The AO technique offers a solution that is locally optimal \cite{bezdek2003convergence} (See Remark \ref{rem_3}).

Thus, we divide \eqref{P2} into three sub-problems. For each one, we optimize \eqref{P2} for the associated variable(s) while keeping the other optimization variables fixed.
The result then feeds into the next sub-problem. This block optimization iterates until the objective function converges. In the first sub-problem, with constant transmit beamformers and reflection coefficients, we optimize received beamformers using \eqref{P1:sinr_tag}. Next, we fix the received beamformers and reflection coefficients to optimize transmit beamformers $\qw$ and $\qS$, navigating the non-convex constraints in \eqref{P2} using the semidefinite relaxation (SDR) method \cite{so2007approximating, Qingqing}. We handle the last sub-problem focused on optimizing reflection coefficients with a novel slack-optimization approach.

Although the original AO approach suggests that the same objective function be optimized over alternative blocks of variables \cite{bezdek2003convergence}, that is not the case here. In our case, the first and third sub-problems are independent of the original objective. These two hence are feasibility problems, where the primary goal is to find a feasible solution that satisfies a set of constraints. Feasibility problems focus solely on finding a point that meets the specified constraints, without necessarily optimizing any objective. Nonetheless, we transform these into optimization problems with explicit objectives to achieve more efficient solutions without compromising the original problem. Consequently, our approach may yield a considerably efficient solution \cite{Qingqing}.

\subsection{Sub-Problem 1: Optimization Over $\qu_k$ }
For given  $\{\qw, \qS,  \{\alpha_k\}_{k\in \mathcal{K}}\}$, problem (P2) becomes a feasibility problem for  receiver beamforming, $\qu_k$. This is because the goal of \eqref{P2}, i.e., the BS transmit power minimization, is independent of $\qu_k$. Any feasible value of $\qu_k$ that satisfies the constraints \eqref{P1:sinr_tag} and \eqref{P1:uk} can thus be a solution.

Although $\qu_k$ might not directly impact for reducing BS transmit power, the sensing  SINR at the BS for each tag depends on the appropriate choice of $\qu_k$. Therefore, we take an approach that seeks to maximize each tag's sensing SINR. This tactic serves a dual purpose: it guarantees that the sensing performance criteria are fulfilled and indirectly supports the overarching objective of transmit power reduction. This is because ensuring high SINR for tag signals can potentially alleviate the need for higher transmit power to overcome poor reception, thus aligning with our power minimization strategy. By optimizing $\qu_k$ to maximize the sensing SINR, we improve power minimization in the subsequent AO  steps \cite{Stanczak2008book, Wan2016}.

Utilizing the unique structure of the sensing SINR for each tag \eqref{eqn_tag_sen_SINR},  we transform this sub-problem into a generalized Rayleigh quotient optimization problem, which has a direct closed-form solution \cite{Stanczak2008book, Wan2016}. Consequently, we obtain the following optimization problem:
\begin{subequations}
\begin{align}\label{P3}
\text{(P3)}:~&   \max_{\qu_k}  \quad   \frac{\alpha_k \qu_k^{\rm{H}} \qG_k \qR_x \qG^{\rm{H}}_k \qu_k}{\qu^{\rm{H}}_k\left(\sum_{i\in \mathcal{K}_k} \alpha_i \qG_i \qR_x \qG_i^{\rm{H}} +  \sigma^2_k \qI_N \right)\qu_k}, \\
\text{s.t} \quad & \|\qu_k\|^2=1,~ \forall k.
\end{align}
\end{subequations}
The  objective function in \eqref{P3} can be restated as  the  following optimization problem: 
\begin{subequations}
\begin{align}\label{P4}
\text{(P4)}:~&  \max_{\qu_k}   \frac{\qu_k^{\rm{H}} \tilde{\qG}_k \tilde{\qG}^{\rm{H}}_k \qu_k}{\qu_k^{\rm{H}} \qQ \qu_k},\quad
\text{s.t} \quad \|\qu_k\|^2=1,~ \forall k,
\end{align}
\end{subequations}
where $\tilde{\qG}_k=\sqrt{\alpha_k} \qG_k (\qw+\qs)$ and $\qQ =  \sum_{i\in \mathcal{K}_k} \alpha_i \qG_i \qR_x \qG_i^{\rm{H}} +  \sigma^2_k \qI_N$. Problem (P4) in \eqref{P4} is  a generalized Rayleigh ratio quotient problem \cite{Stanczak2008book, Wan2016}. When the  transmit beamformers and  reflection coefficients are fixed, the  optimal received beamformer is thus given by
\begin{align}\label{opt:uk}
\qu_k^* = \frac{\qQ^{-1} \tilde{\qG}_k}{\|\qQ^{-1} \tilde{\qG}_k\|},~ \forall k,
\end{align}
which is an MMSE filter \cite{Stanczak2008book, Wan2016}.

\setcounter{mycounter}{\value{equation}}
\begin{figure*}[!t]
\addtocounter{equation}{1}
\begin{eqnarray} \label{P6}
\text{(P6)}:&&\!\!\!\!\!\!\!\! \underset{\qW, \qS} {\text{minimize}}\:\:  \tr(\qW) + \tr(\qS)   \nonumber \\
\text{s.t}  &&\!\!\!\!\!\!\!\! \Upsilon_k^{\thr}{\qu^{\rm{H}}_k\left(\sum_{i\in \mathcal{K}_k} \alpha_i  \left( \tr(\qG_i^{\rm{H}}\qG_i\qW) + \tr(\qG_i^{\rm{H}}\qG_i\qS) \right)  +  \sigma^2_k \qI_N \right)\qu_k} - {\alpha_k \qu_k^{\rm{H}} \left( \tr(\qG_i^{\rm{H}}\qG_i\qW) + \tr(\qG_i^{\rm{H}}\qG_i\qS) \right) \qu_k }  \leq 0, \forall k    ,\nonumber\\ 
&&\!\!\!\!\!\!\!\!\frac{\tr(\qf\qf^{\rm{H}}\qW)}{\Gamma_u^{\thr}(1+\Gamma_k^{\thr})}-  {\sum_{i\in \mathcal{K}_k}\alpha_i \left(\tr(\qh_i\qh_i^{\rm{H}}\qW) + \tr(\qh_i\qh_i^{\rm{H}}\qS)\right)  }  \geq \sigma^2 , \forall k, \nonumber\\ 
&&\!\!\!\!\!\!\!\! P_{\rm{th}}-(1-\alpha_k)(\tr(\qg_{f,k}\qg^{\rm{H}}_{f,k}\qW) +\tr(\qg^{\rm{H}}_{f,k}\qg_{f,k}\qS))\leq 0, \forall k.
\end{eqnarray}	
\hrulefill
\end{figure*}
\setcounter{equation}{\value{mycounter}}

\subsection{Sub-Problem 2: Optimization Over $\qw$ and $\qS$ } 
For given $\{\{\qu_k\}_{k\in \mathcal{K}},  \{\alpha_k\}_{k\in \mathcal{K}}\}$, problem (P2) can be reformulated as the following equivalent problem:
\begin{subequations}
\begin{align}\label{P5}
\text{(P5)}:~& \min_{\qw, \qS} \quad   \|\qw\|^2 + \tr(\qS), \quad \\
&\text{s.t} \quad  \eqref{P2:sinr_user},~ \eqref{P1:sinr_tag},~\eqref{P1:EH}.
\end{align}
\end{subequations}
Utilizing the SDR method, we can adeptly address problem \eqref{P5} \cite{so2007approximating, Qingqing}. We introduce the matrix definition $ \mathbf{W} = \mathbf{w} \mathbf{w}^{\rm{H}}$. By exploiting that $\mathbf{W}$ is semidefinite and has   $\text{Rank}(\mathbf{W}) = 1$, problem \eqref{P5} can be recast as \eqref{P6} where the rank one constraint is dropped to relax the problem.

Note that the relaxation of the rank in \eqref{P6} represents a conventional semi-definite programming (SDP)  problem \cite{boyd2004convex}, which can be tackled using the CVX tool \cite{boyd2004convex, grant2014cvx}. 
Let the solution to this relaxed SDR problem be $\qW^*$  with the  eigenvalue decomposition: $\qW^* = \qU \boldsymbol{\Sigma} \qU^{\rm{H}}$ where $\qU$ is a unitary matrix and $\boldsymbol{\Sigma}  = \text{diag}(\lambda_1, \dots, \lambda_{M})$ is a diagonal matrix, both sized $M \times M$. If $\qW^*$ is rank one, the optimal transmit beamformer, $\qw^*$, is the eigenvector for the maximum eigenvalue. Otherwise, to account for the relaxed rank-one constraint, we utilize the Gaussian randomization \cite{Qingqing}. Specifically, we compute a  solution for \eqref{P5} as $\bar{\qW} = \qU \boldsymbol{\Sigma} ^{1/2} \qr$, with $\qr \in \mathcal{CN}(0, \qI_{M})$. We do this for $10^5$ times and select the best.
These numerous random realizations of  $\qr$ with the SDR technique ensure a $\frac{\pi}{4}$-approximation to the optimal value of \eqref{P5} \cite{so2007approximating, Qingqing}.

\subsection{Sub-Problem 3: Optimization Over $\alpha_k$ } 
This one focuses on optimizing each tag's reflection coefficient ($\alpha_k$). By isolating the variables and constraints relevant to this sub-problem, we transform the original optimization problem \eqref{P1}  into a feasibility problem as follows:
\addtocounter{equation}{1}
\begin{subequations}
\begin{align} \label{P7_1}
\text{(P7)}:~& {\rm{find}} \quad \alpha_k  \\
\text{s.t} \quad & \eqref{P1:sinr_tag}-\eqref{P1:EH},~\eqref{P1:alpha}, \label{eqn_P7_const}
\end{align}
\end{subequations}
where  any $\alpha_k$ that satisfies (P7)  is considered a feasible solution. However, the feasible solution yielded from (P7)  does not guarantee that the constraints are satisfied with equality \cite{Qingqing}. Hence, to achieve a better solution, we further transform this into an optimization problem with an explicit objective to obtain generally more efficient reflection coefficients to reduce the transmit power \cite{Qingqing}. The rationale is that for the transmit beamforming optimization problem, i.e., (P6) \eqref{P6}, all SINR and EH constraints are active at the optimal solution; thus, optimizing the reflection coefficient to force the tag SINR and EH to be greater than the targeted values in (P8) directly leads to a reduction in transmit power in (P6) \cite{Qingqing}. Following the slack variable optimization technique in \cite{Shayan_Zargari, Razaviyayn2013, Qingqing}, we can introduce two new slack variables, $t_1$ and $t_2$, which represent the “SINR residual” and “EH residual” to further optimize the SINR and EH margins while satisfying constraint \eqref{eqn_P7_const}.
We then propose solving the following sub-problem: 
\begin{subequations}
\begin{align} \label{P7}
\text{(P8)}:~& \min_{\alpha_k, t_1, t_2}  \quad  \lambda_1 t_1 + \lambda_2 t_2  \\
\text{s.t} \quad & {\alpha_k \qu_k^{\rm{H}} \qG_k \qR_x \qG^{\rm{H}}_k \qu_k }  \geq \\ &   \Upsilon_k^{\thr}{\qu^{\rm{H}}_k\left(\sum_{i\in \mathcal{K}_k} \alpha_i \qG_i \qR_x \qG_i^{\rm{H}} +  \sigma^2_k \qI_N \right)\qu_k}+ t_1, \forall k    ,\nonumber \label{P7_SINR}\\ 
&  \frac{|\mathbf{f}^{\rm{H}} \mathbf{w} |^2}{\Gamma_u^{\thr}(1+\Gamma_k^{\thr})}-  \sum_{i\in \mathcal{K}_k} \alpha_i \left(|\qh_i^{\rm{H}} \qw|^2+\qh^{\rm{H}}_i \qS \qh_i \right)  \geq \sigma^2 , \forall k    ,\\ 
& (1-\alpha_k)p_{k}^{\rm{in}}  \geq  {\Phi^{-1}(p_b)} + t_2, \forall k,   \\
&  \eqref{P1:alpha}, \label{P7_alpha}
\end{align}
\end{subequations}
where $\lambda_1$ and $\lambda_2$ are positive constants. Problem \eqref{P7} is convex and thus can be efficiently solved by solvers such as CVX  \cite{grant2014cvx}. Although  \eqref{P7_1} and \eqref{P7} share the same feasible set, the introduction of slack variables in \eqref{P7} converts strict constraints into adjustable ones with a definable margin. This facilitates the convergence process by setting a more tangible minimization goal and aligns well with the convergence strategies of iterative solvers like CVX due to the explicit objective guiding the solution path \cite{Shayan_Zargari,Qingqing}.

\begin{algorithm}[!t]
\caption{AO Algorithm}
\begin{algorithmic}[1]
\label{alg:AO:maxmin}
\STATE \textbf{Input}: Set the iteration counter $t = 0$, the convergence tolerance $\epsilon > 0$, initial feasible solution $\{\alpha_k\}_{k\in \mathcal{K}}, \qw, \qS$. Initialize the objective function value $F^{(0)} = 0$.  
\WHILE{ $ \frac{F^{(t+1)} - F^{(t)}}{F^{(t+1)}} \geq \epsilon$}
\STATE Solve \eqref{opt:uk} for the  received beamformer, $\qu_k^{(t+1)}$.
\STATE Solve \eqref{P5} to obtain transmit beamformers, $\{\qw^{(t+1)},\; \qs^{(t+1)}\}$ by recovering a rank-one solution via  Gaussian randomization 
\STATE Solve \eqref{P7} for the  reflection coefficients, $\alpha_k^{(t+1)}$. 
\STATE Calculate the objective function value $F^{(t+1)}$.
\STATE Set $t\leftarrow t+1$;
\ENDWHILE
\STATE \textbf{Output}: Optimal solutions $\mathcal{A}^*$.
\end{algorithmic}
\end{algorithm}

Our algorithm to solve \eqref{P1} is presented in Algorithm \ref{alg:AO:maxmin}. It starts by initializing $\{\{\alpha_k\}_{k\in \mathcal{K}}, \qw, \qS\}$ to random  feasible values
and, in every iteration, refines the values of received/transmit beamformers and reflection coefficients until the normalized improvement of the total transmit power is smaller than $\epsilon=\num{1e-3}$.

\begin{rem}\label{rem_3}
Each sub-problem within the AO algorithm is designed to have a local solution. The convergence of the AO algorithm to a local minimum or stationary point is generally guaranteed for a wide range of problems, provided that certain conditions are met \cite{bezdek2003convergence}. These conditions may include the objective function being lower-bounded and having certain smoothness properties. Specifically, as long as the individual sub-problems converge, the overall optimization also converges \cite{bezdek2003convergence}. By exploiting that insight, we use the SDR and slack-optimization methods to solve  $\{\mathbf{w}, \mathbf{s}\}$ and $\{\alpha_k\}_{k\in \mathcal{K}}$, respectively, whereas $\{\qu_k\}_{k\in \mathcal{K}}$ is obtained as a closed-form solution applying the Rayleigh ratio quotient approach. SDR and slack-optimization are well-developed approaches with provable convergence \cite{so2007approximating, Qingqing, Razaviyayn2013}, ensuring the convergence of our proposed AO algorithm. This claim is also validated by our simulations (Fig.~\ref{conv_fig}).
\end{rem}

\begin{theorem}
Algorithm \ref{alg:AO:maxmin} iterations yield a non-increasing sequence of objective values with guaranteed convergence.
\end{theorem}

\begin{proof}
   Please see  Appendix \ref{Theorem_1}.
\end{proof}

\subsection{Computational Complexity of the Proposed Algorithm}
 This is analyzed for the three sub-problems. 

\subsubsection{Optimization over $\qu_k$}
During this phase, the optimal received beamformers are obtained using the Rayleigh quotient. Computing the inverse of the matrix $\qQ$ requires $\mathcal{O}(N^3)$. Additionally, the MMSE filter for the $K$ tags (as shown in \eqref{opt:uk}) adds complexity of $\mathcal{O}(KN^2)$. Therefore, the total complexity for this section is $\mathcal{O}(KN^2 + N^3)$.

\subsubsection{Optimization over $\qw$ and $\qS$}
Note that the interior-point method can solve sub-problems based on the SDP. According to \cite[Th. 3.12]{polik2010interior}, the order of complexity for a SDP problem with $m$ SDP constraints which includes a $n \times n$ positive semi-definite (PSD) matrix is given by $\mathcal{O}\left( \sqrt{n} \log\left(\frac{1}{\epsilon}\right) (mn^3 + m^2n^2 + m^3) \right)$, where $\epsilon > 0$ is the solution accuracy. For problem \eqref{P5}, with $n = M$ and $m = 3K + 2$, the approximate computational complexity for solving \eqref{P5} can be written as
$\mathcal{O}\left( KM^3 \sqrt{M} \log\left(\frac{1}{\epsilon}\right) \right)$.

\subsubsection{Optimization over $\alpha_k$}
Utilizing CVX, the optimization leverages the DC and interior point methods. The iterations required for convergence can be expressed as $\frac{\left({\log(C)}/{t^0\delta}\right)}{\log \epsilon}$. Here, $C$ denotes the overall number of constraints. The term $t^0$ signifies the initial approximation for the interior point method's accuracy. The stopping criterion is $0 < \delta \ll 1$ \cite{boyd2004convex}.

\subsubsection{Algorithm \ref{alg:AO:maxmin}}
The computational complexity of each iteration of Algorithm \ref{alg:AO:maxmin} is asymptotically equal to $ \mathcal{O}\left(\!I \!\left(KN^2 \!+ \!N^3 \!+ \! KM^3 \sqrt{M} \log\left(\frac{1}{\epsilon}\right) \!+ \!\frac{\left({\log(C)}/{t^0\delta}\right)}{\log \epsilon}\right)\!\right)\!$, where $I$ is the required number of iterations for the outer algorithm to converge. Despite its higher-order polynomial time complexity, the proposed algorithm demonstrates commendable real-world performance for datasets up to a particular size, predominantly when $N$ and $M$ are maintained below a defined threshold. For extensive datasets, embracing optimization strategies, like parallel processing, can significantly improve the performance \cite{leiserson2010parallel}.

\section{Simulation Results}
We next present simulation results for assessing the performances of the proposed ISABC network and the AO algorithm. 

\subsection{Simulation Setup and Parameters}

The 3GPP urban micro (UMi) model is adopted to model the path-loss \{$\zeta_{a}, \zeta_b$\}  with $f_c =  \qty{3}{\GHz}$ operating frequency \cite[Table B.1.2.1]{3GPP2010}. 

{The considered  {\abc}  leverages the existing RF signals for data transmission. The carrier frequency choice is thus compatible with existing communication infrastructure and aligns with the future-forward vision of {\abc} systems for 5G/6G spectrum utilization trends.}
We model the AWGN variance, $\sigma^2$, as $\sigma^2=10\log_{10}(N_0 B N_f)$ \qty{}{\dB m}, where $N_0=\qty{-174}{\dB m/\Hz}$, $B$ is the bandwidth, and $N_f$ is the noise figure. Unless otherwise specified, Table \ref{table-notations} summarizes the simulation parameters. All simulations are evaluated for \num{1e3} iterations. 

Considering a smart home use case, we place the BS and the mobile reader at $\{0,0\}$ and $\{12,0\}$, respectively, while the tags are randomly distributed within a circle centered at $\{6,-4\}$ with a radius of \qty{3}{\m} \cite{Galappaththige2023, Xu2021NOMA}.

\begin{table}[t]
\renewcommand{\arraystretch}{1.05}
\centering
\caption{Simulation parameters.}
\label{table-notations}
\begin{tabular}{c c c c}    
\hline
\rowcolor{LightGray}
\textbf{Parameter}& \textbf{Value} & \textbf{Parameter}& \textbf{Value}\\  \hline
$f_c$ & \qty{3}{\GHz}  & $\Gamma_u^{\thr}$  & \qty{1}{bps/\Hz}  \\ 
$B$ & \qty{10}{\MHz}  &  $\Gamma_k^{\thr}$ & \qty{1}{bps/\Hz} \\ 
$N_f$ & \qty{10}{\dB}  &  $p_b$  & \qty{-20}{\dB m} \\ 
$M = N $ & \num{8}  &  $M_{\rm{NL}}$  & \qty{20e-3}{\watt}   \\ 
$K$ & \num{3}  &  $a_{\rm{NL}}$ & \num{6400}  \\ 
$\Upsilon_k^{\thr}$ & \qty{1}{bps/\Hz}  &  $b_{\rm{NL}}$ & \num{0.003}  \\ \hline
\end{tabular}
\end{table}

\subsection{Benchmark Schemes}
We denote the proposed ISABC system with passive backscatter tags as `ISABC-P'. For comparative  evaluation purposes,  we consider the following benchmarks:
\subsubsection{Convectional ISAC}\label{sec_con_ISAC}
The first benchmark `ISAC' is conventional ISAC with an FD BS \cite{Zhitong2022, Zhenyao2023}. This model does not include passive tags but conventional radar targets that only reflect incident signals and do not send data to the user. However, the sensing waveform does interfere with the detection process at the user. This sensing waveform is assumed to be perfectly cancelled at the user and the BS \cite{Zhitong2022}. However, the reflected signals by the targets cause interference for the user. 

\subsubsection{Conventional \bc}
This benchmark (`\bc') comprises one user and multiple backscatter tags, and the BS does not perform sensing, i.e., $\mathbf{x} = \mathbf{w}x_d$. The tags perform EH  to power their internal functions while sending data to the user by backscattering the BS RF signals. Thus, this benchmark helps to evaluate the cost of incorporating sensing functions on \bc communication performance.

\subsubsection{Communication-only scheme} 
This benchmark (legend `Com-only') assesses the system's core communication capacity and resilience by focusing solely on a single-user scenario. It establishes a baseline performance metric by isolating communication from sensing and backscattering. Deviations from this metric reveal the impact of added functionalities like sensing or backscattering.

\subsubsection{Sensing-only scheme}
This benchmark (legend `Sensing-only') focuses on a sensing-centric system considering EH, excluding primary communication and \bc. It helps establish a baseline for assessing trade-offs in integrated systems and highlights the system's raw sensing performance, particularly in cases prioritizing sensing over sporadic or secondary communications.

\begin{figure}[!t]\vspace{-0mm}	
\centering
\fontsize{14}{14}\selectfont 
    \resizebox{.55\totalheight}{!}{
%
%
\definecolor{mycolor1}{rgb}{0.00000,0.44700,0.74100}%
\definecolor{mycolor2}{rgb}{0.85000,0.32500,0.09800}%
\definecolor{mycolor3}{rgb}{0.92900,0.69400,0.12500}%
\definecolor{mycolor4}{rgb}{0.49400,0.18400,0.55600}%
\begin{tikzpicture}

\begin{axis}[%
width=5.877in,
height=4.917in,
at={(0.986in,0.664in)},
scale only axis,
xmin=1,
xmax=10,
xtick={ 1,  2,  3,  4,  5,  6,  7,  8,  9, 10},
xlabel style={font=\color{white!15!black}},
xlabel={Number of iterations},
ymin=48,
ymax=62,
ytick={48, 50, ..., 62},
ylabel style={font=\color{white!15!black}},
ylabel={Transmission power (dBm)},
axis background/.style={fill=white},
xmajorgrids,
ymajorgrids,
legend style={legend cell align=left, align=left, draw=white!15!black}
]
\addplot [color=mycolor1, line width=1.5pt, mark size=3.0pt, mark=o, mark options={solid, mycolor1}]
  table[row sep=crcr]{%
1	60.1132\\
2	50.5925\\
3	50.5878\\
4	50.5878\\
5	50.5878\\
6	50.5878\\
7	50.5878\\
8	50.5878\\
9	50.5878\\
10	50.5878\\
};
\addlegendentry{ISABC-P, $K = 1, M = 4$}

\addplot [color=mycolor2, line width=1.5pt, mark size=2.5pt, mark=square, mark options={solid, mycolor2}]
  table[row sep=crcr]{%
1	55.0544\\
2	50.9292\\
3	50.9361\\
4	50.9319\\
5	50.9333\\
6	50.9326\\
7	50.9329\\
8	50.9327\\
9	50.9328\\
10	50.9327\\
};
\addlegendentry{ISABC-P, $K = 3, M = 4$}

\addplot [color=mycolor3, line width=1.5pt, mark size=3.5pt, mark=diamond, mark options={solid, mycolor3}]
  table[row sep=crcr]{%
1	52.8479\\
2	48.4015\\
3	48.4063\\
4	48.4062\\
5	48.4062\\
6	48.4062\\
7	48.4062\\
8	48.4062\\
9	48.4062\\
10	48.4062\\
};
\addlegendentry{ISABC-P, $K = 1, M = 8$}

\addplot [color=mycolor4, line width=1.5pt, mark size=3.5pt,  mark=triangle, mark options={solid, mycolor4}]
  table[row sep=crcr]{%
1	57.1784\\
2	48.7259\\
3	48.7259\\
4	48.7257\\
5	48.7259\\
6	48.7257\\
7	48.7259\\
8	48.7257\\
9	48.7259\\
10	48.7257\\
};
\addlegendentry{ISABC-P, $K = 3, M = 8$}

\end{axis}

\begin{axis}[%
width=7.583in,
height=6.033in,
at={(0in,0in)},
scale only axis,
xmin=0,
xmax=1,
ymin=0,
ymax=1,
axis line style={draw=none},
ticks=none,
axis x line*=bottom,
axis y line*=left
]
\end{axis}
\end{tikzpicture}%
}\vspace{-0mm}
    \caption{Convergence rate.}
	\label{conv_fig} \vspace{-0mm}
\end{figure}
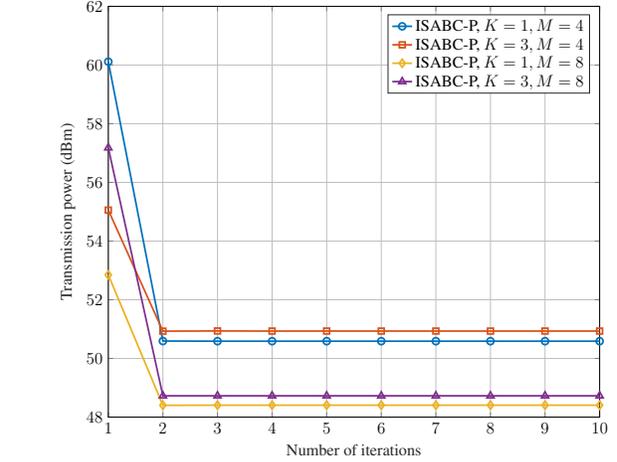

\subsubsection{ISABC with active tags}
Benchmark `ISABC-A' evaluates the performance of the ISABC system with active tags (battery-powered), eliminating the need for  EH requirements at the tags and thus reducing the BS transmit power. This benchmark thus establishes a baseline to gauge the cost to ISABC  of ensuring EH at the tags. 

Our Algorithm \ref{alg:AO:maxmin} accommodates all these benchmarks as special cases. Additionally, we explore three more benchmarks based on the tag's reflection coefficient and  BS beamformers. 
\begin{enumerate}
\item[6)] \textit{Random reflection coefficients:} 
Tags might reflect signals without a set pattern, leading to random reflection coefficients. Such randomness can arise from environmental changes, tag characteristics, or varied communication protocols \cite{Ardakani2022}. In contrast, optimizing reflection coefficients helps select the proper impedances, enhancing rates and ranges. Yet, this optimization requires more computational resources from the reader. For certain cost-sensitive applications, such an approach might not be economical. Thus, this baseline seeks to gauge the tradeoff between optimizing the tag reflection coefficients and not optimizing.

\end{enumerate}

The next two benchmarks are motivated by the following considerations. Algorithm \ref{alg:AO:maxmin} aims to boost the SINR, focusing especially on sub-problems (P4). The outcome of this strategy is the MMSE filter, as shown in \eqref{opt:uk}. This algorithm iteratively refines the MMSE filter, drawing insights from its other two sub-problems. One might consider omitting a sub-problem to streamline this process, possibly by opting for simpler solutions for each $\qu_k,~\forall k\in\mathcal{K}$. Many receivers have leaned towards match filter (MF) and zero-forcing (ZF) combiners due to their simplicity \cite{Azzam2021, Azar_Hakimi}. However, MF struggles with multi-tag interference, while ZF is best suited for high SNR regions because of its sensitivity to noise. To better grasp these concepts, we briefly explore MF and ZF beamformers.
\begin{enumerate} 
\item[7)]\textit{MF beamformer:}  $\qu_{\text{MF}} = \qH_b$, where this leverages the CSI to amplify the received signal's strength. Though advantageous in several scenarios, the MF beamformer's inability to eliminate multi-tag interference is a limitation.
\item[8)] \textit{ZF beamformer:} $\qu_{\rm{ZF}} = \qH_b\left(\qH_b^{\rm{H}} \qH_b\right)^{-1}$, which strives to obliterate interference at non-intended receivers by distinctively utilizing the CSI. It effectively manages interference, but its susceptibility to noise, especially in low-SNR environments,  is challenging.
\end{enumerate}
Elaborating further, $\qH_b\in \mathbb{C}^{N\times K}$ symbolizes the backward channel matrix. Within this matrix, each $k$-th column's vector is articulated as $\mathbf{g}_{b,k}$, $\forall k \in \mathcal{K}$. Interestingly, these combiners rely solely on CSI, in stark contrast to the MMSE filter \eqref{opt:uk}, which is influenced by tag reflection coefficients and the precoder. Utilizing the aforementioned combiners might lead to inferior performance concerning transmit power. Nevertheless, this paves the way for simplifying the three-stage AO approach into a more concise two-stage algorithm, which inherently means a faster algorithmic execution.

\begin{figure}[!t]\vspace{-0mm}	
\centering
\fontsize{14}{14}\selectfont 
    \resizebox{.55\totalheight}{!}{\input{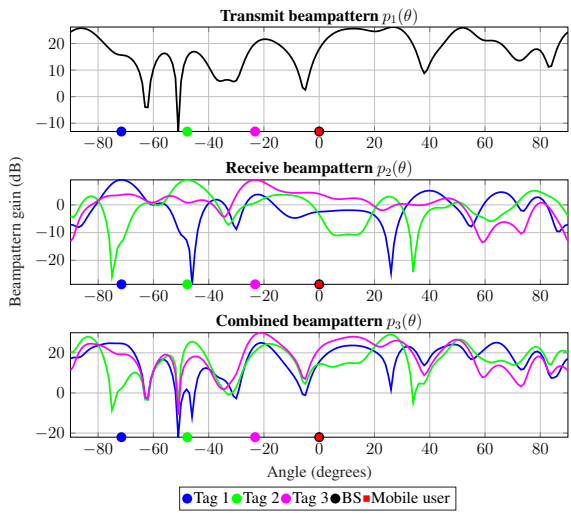}}\vspace{-0mm}
    \caption{Beampattern regarding radar functionality of Algorithm \ref{alg:AO:maxmin}.}
	\label{beam_pro} \vspace{-0mm}
\end{figure}

\begin{figure}[!t]\vspace{-0mm}	
\centering
\fontsize{14}{14}\selectfont 
    \resizebox{.55\totalheight}{!}{\input{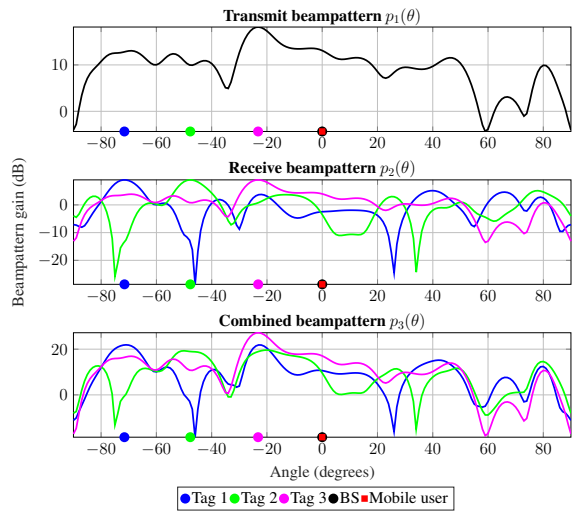}}\vspace{-0mm}
    \caption{Beampattern regarding communication functionality of only communication.}
	\label{beam_sensing} \vspace{-0mm}
\end{figure}

\subsection{Convergence Rate of Algorithm \ref{alg:AO:maxmin}}
This outputs the optimal BS received beamformer, $\{\qu_k\}_{k\in \mathcal{K}}$, BS transmit beamforming for communication and sensing, $\qw$ and $ \qS$, respectively, and the tag reflection coefficients, $\{\alpha_k\}_{k\in \mathcal{K}}$, for a given ISABC system setup within several iterations. The BS transmit power stabilizes at the end of several iterations, which indicates convergence. To measure it, Fig.~\ref{conv_fig} plots the BS transmit power as a function of the number of iterations. It decreases with each iteration until it reaches a fixed value after approximately three iterations, regardless of  $M$ or $K$. This indicates a rapid convergence and verifies the efficacy of the proposed algorithm. Overall, Algorithm \ref{alg:AO:maxmin} only requires three iterations to achieve satisfactory performance with any system configuration, resulting in minor performance enhancements beyond three iterations.

\subsection{Beampattern Gains}
Modern radar functionality has evolved to harness the power of beamforming, i.e.,  directly transmit and receive beams in specific directions. Algorithm \ref{alg:AO:maxmin} serves this function, guiding the formation and steering of these beams. Beamforming in radar systems converges signals from an array of antennas, crafting a directed “beam” or “lobe” \cite{Liu2018Radar, Zhenyao2023}. Intriguingly, this beam can be electronically steered while the antennas remain stationary. This electronic steering capability amplifies signal quality, boosts backscatter tag detection, and significantly minimizes potential interference \cite{Liu2018Radar, Zhenyao2023}. 
 
\begin{figure}[!t]\vspace{-0mm}	
\centering
\fontsize{14}{14}\selectfont 
    \resizebox{.55\totalheight}{!}{
%
%
\definecolor{mycolor1}{rgb}{0.00000,0.44700,0.74100}%
\definecolor{mycolor2}{rgb}{0.85000,0.32500,0.09800}%
\definecolor{mycolor3}{rgb}{0.92900,0.69400,0.12500}%
\definecolor{mycolor4}{rgb}{0.49400,0.18400,0.55600}%
\begin{tikzpicture}

\begin{axis}[%
width=5.877in,
height=4.917in,
at={(0.986in,0.664in)},
scale only axis,
bar shift auto,
xmin=0.509090909090909,
xmax=11.4909090909091,
xlabel style={font=\color{white!15!black}},
xlabel={Number of tags, $K$},
ymin=0,
ymax=18,
ylabel style={font=\color{white!15!black}},
ylabel={Average running time (sec)},
axis background/.style={fill=white},
xmajorgrids,
ymajorgrids,
legend style={at={(0.03,0.97)}, anchor=north west, legend cell align=left, align=left, draw=white!15!black}
]
\addplot[ybar, bar width=0.145, fill=mycolor1, draw=black, area legend] table[row sep=crcr] {%
1	10.2452\\
2	10.5289\\
3	10.614\\
4	10.7517\\
5	11.6698\\
6	13.7508\\
7	14.1954\\
8	15.0886\\
9	16.4131\\
10	16.8123\\
11	17.2123\\
};
\addplot[forget plot, color=white!15!black] table[row sep=crcr] {%
0.509090909090909	0\\
11.4909090909091	0\\
};
\addlegendentry{ISABC-P}

\addplot[ybar, bar width=0.145, fill=mycolor2, draw=black, area legend] table[row sep=crcr] {%
1	1.630996479\\
2	2.163523408\\
3	2.471518008\\
4	1.52248388\\
5	1.265699546\\
6	1.684177543\\
7	1.963964733\\
8	1.960566674\\
9	2.340978285\\
10	2.235436698\\
11	3.821483536\\
};
\addplot[forget plot, color=white!15!black] table[row sep=crcr] {%
0.509090909090909	0\\
11.4909090909091	0\\
};
\addlegendentry{ISABC-A}

\addplot[ybar, bar width=0.145, fill=mycolor3, draw=black, area legend] table[row sep=crcr] {%
1	7.75019717866667\\
2	8.42937812564102\\
3	8.64364007894737\\
4	9.3308527390625\\
5	9.60498249\\
6	9.633957168\\
7	9.75758517894737\\
8	9.97609098666667\\
9	10.1460916958333\\
10	10.6854965703125\\
11	10.76599795\\
};
\addplot[forget plot, color=white!15!black] table[row sep=crcr] {%
0.509090909090909	0\\
11.4909090909091	0\\
};
\addlegendentry{Sensing-only}

\addplot[ybar, bar width=0.145, fill=mycolor4, draw=black, area legend] table[row sep=crcr] {%
1	1.312243857\\
2	1.319299281\\
3	1.407297041\\
4	1.463662292\\
5	1.482405893\\
6	1.5776062040404\\
7	1.653992808\\
8	1.78503151546392\\
9	1.93043212159091\\
10	2.13086855064935\\
11	2.25634646607143\\
};
\addplot[forget plot, color=white!15!black] table[row sep=crcr] {%
0.509090909090909	0\\
11.4909090909091	0\\
};
\addlegendentry{ISAC}

\end{axis}

\begin{axis}[%
width=7.583in,
height=6.033in,
at={(0in,0in)},
scale only axis,
xmin=0,
xmax=1,
ymin=0,
ymax=1,
axis line style={draw=none},
ticks=none,
axis x line*=bottom,
axis y line*=left
]
\end{axis}
\end{tikzpicture}%
}\vspace{-0mm}
    \caption{The running time versus the number of tags, $K$. }
	\label{runn_time} \vspace{-5mm}
\end{figure}
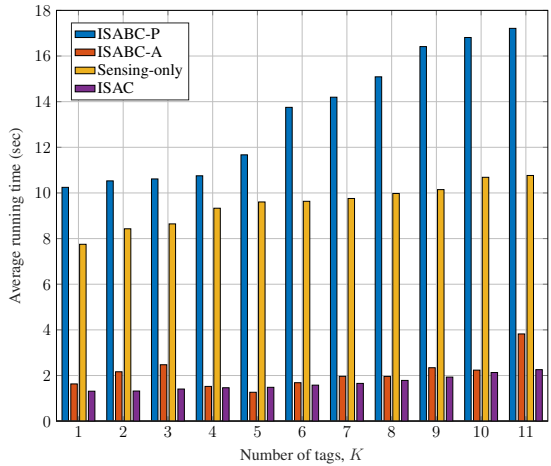

The transmit signal, representing the outward-projected energy, is critical in efficiently illuminating the radar's targets. Meanwhile, the received beamformer, normalized as $\|{\qu^*_k}\| = 1$, is optimized for clear reception, capturing the echoes or reflections off the backscatter tags. Three pivotal beampatterns arise from these components:
\begin{subequations}
\begin{align} 
p_1(\theta) &= \left| \mathbf{b}^{\rm{H}}(\theta_k) \qx^* \right|^2, \label{eq:30} \\ p_2(\theta) &= \left| (\qu_k^*)^{\rm{H}} \mathbf{b}(\theta_k) \right|^2, \label{eq:31} \\
p_3(\theta) &= \left| (\qu_k^*)^{\rm{H}} \mathbf{b}(\theta_k) \mathbf{b}^{\rm{H}}(\theta_k) \qx^* \right|^2. \label{eq:32} 
\end{align}
\end{subequations}
First, \eqref{eq:30} illustrates how the transmitted energy disperses as a function of angle $\theta$. Second, \eqref{eq:31} encapsulates the sensitivity of the radar system across different angles during the reception of reflected energy. Finally, \eqref{eq:32} offers a combined representation, integrating the effects of transmission and subsequent reflection processing.

Figures \ref{beam_pro} and \ref{beam_sensing}  make it easier to discern the nuances and effectiveness of beamforming algorithms. 

\begin{figure}[!t] 
\centering
\fontsize{15}{15}\selectfont 
\begin{subfigure}[b]{0.45\textwidth}
\centering
\fontsize{14}{14}\selectfont 
    \resizebox{.55\totalheight}{!}{
%
%
\definecolor{mycolor1}{rgb}{0.00000,0.44700,0.74100}%
\definecolor{mycolor2}{rgb}{0.85000,0.32500,0.09800}%
\definecolor{mycolor3}{rgb}{0.92900,0.69400,0.12500}%
\definecolor{mycolor4}{rgb}{0.49400,0.18400,0.55600}%
\definecolor{mycolor5}{rgb}{0.46600,0.67400,0.18800}%
\definecolor{mycolor6}{rgb}{0.30100,0.74500,0.93300}%
\begin{tikzpicture}

\begin{axis}[%
width=5.877in,
height=4.917in,
at={(0.986in,0.664in)},
scale only axis,
xmin=5,
xmax=30,
xtick={5, 10, ..., 30},
xlabel style={font=\color{white!15!black}},
xlabel={Number of transmit and recieved antenna at BS, $M = N$},
ymin=20,
ymax=68,
ytick={20, 25, ..., 65},
ylabel style={font=\color{white!15!black}},
ylabel={Transmission power (dBm)},
axis background/.style={fill=white},
xmajorgrids,
ymajorgrids,
legend style={at={(0.975,0.85)}, legend cell align=left, align=left, draw=white!15!black}
]
\addplot [color=mycolor1, line width=1.5pt, mark size=3.0pt, mark=o, mark options={solid, mycolor1}]
  table[row sep=crcr]{%
5	50.9271\\
10	45.9448\\
15	43.2884\\
20	41.641\\
25	40.585\\
30	39.6431\\
};
\addlegendentry{ISABC-P}

\addplot [color=mycolor2, line width=1.5pt, mark size=4.5pt, mark=x, mark options={solid, mycolor2}]
  table[row sep=crcr]{%
5	49.3199\\
10	44.412\\
15	40.4401\\
20	37.7408\\
25	35.8247\\
30	34.2975\\
};
\addlegendentry{\bc}

\addplot [color=mycolor3, line width=1.5pt, mark size=3.5pt, mark=diamond, mark options={solid, mycolor3}]
  table[row sep=crcr]{%
5	32.1241\\
10	30.0473\\
15	28.9464\\
20	27.7876\\
25	25.5243\\
30	23.9829\\
};
\addlegendentry{Sensing-only}

\addplot [color=mycolor4, line width=1.5pt, mark size=3.5pt, mark=triangle, mark options={solid, mycolor4}]
  table[row sep=crcr]{%
5	50.9271\\
10	45.9448\\
15	43.2884\\
20	41.641\\
25	40.585\\
30	39.6431\\
};
\addlegendentry{ZF}

\addplot [color=mycolor5, line width=1.5pt, mark size=2.5pt, mark=square, mark options={solid, mycolor5}]
  table[row sep=crcr]{%
5	50.9271\\
10	45.9448\\
15	43.2884\\
20	41.641\\
25	40.585\\
30	39.6431\\
};
\addlegendentry{MF}

\addplot [color=mycolor6, line width=1.5pt, mark size=4.0pt, mark=star, mark options={solid, mycolor6}]
  table[row sep=crcr]{%
5	64.9283\\
10	64.8143\\
15	64.5801\\
20	64.2839\\
25	64.1618\\
30	63.3423\\
};
\addlegendentry{Random $\alpha$}

\end{axis}

\begin{axis}[%
width=7.583in,
height=6.033in,
at={(0in,0in)},
scale only axis,
xmin=0,
xmax=1,
ymin=0,
ymax=1,
axis line style={draw=none},
ticks=none,
axis x line*=bottom,
axis y line*=left
]
\end{axis}
\end{tikzpicture}%
}\vspace{-0mm}
    \caption{}
    \label{subfig:antenna_1}
\end{subfigure}
\hfill 
\begin{subfigure}[b]{0.45\textwidth}
\centering
\fontsize{14}{14}\selectfont 
    \resizebox{.55\totalheight}{!}{
%
%
\definecolor{mycolor1}{rgb}{1.00000,0.00000,1.00000}%
\begin{tikzpicture}

\begin{axis}[%
width=5.877in,
height=4.917in,
at={(0.986in,0.664in)},
scale only axis,
xmin=5,
xmax=30,
xtick={5, 10, ..., 30},
xlabel style={font=\color{white!15!black}},
xlabel={Number of transmit and recieved antenna at BS, $M = N$},
ymin=-40,
ymax=-30,
ytick={-40, -38, ..., -30},
ylabel style={font=\color{white!15!black}},
ylabel={Transmission power (dBm)},
axis background/.style={fill=white},
xmajorgrids,
ymajorgrids,
legend style={legend cell align=left, align=left, draw=white!15!black}
]
\addplot [color=mycolor1, line width=1.5pt, mark size=4.0pt, mark=asterisk, mark options={solid, mycolor1}]
  table[row sep=crcr]{%
5	-30.2158\\
10	-33.7622\\
15	-35.6581\\
20	-36.5001\\
25	-37.3298\\
30	-38.0183\\
};
\addlegendentry{ISABC-A}

\addplot [color=black, line width=1.5pt, mark size=2.5pt, mark=square, mark options={solid, black}]
  table[row sep=crcr]{%
5	-30.2958\\
10	-33.8422\\
15	-35.7381\\
20	-36.5801\\
25	-37.4098\\
30	-38.0983\\
};
\addlegendentry{ISAC}

\addplot [color=green, line width=1.5pt, mark size=3.0pt, mark=o, mark options={solid, green}]
  table[row sep=crcr]{%
5	-32.4719\\
10	-35.4795\\
15	-36.8531\\
20	-37.6564\\
25	-38.3677\\
30	-39.0758\\
};
\addlegendentry{Com-only}

\end{axis}

\begin{axis}[%
width=7.583in,
height=6.033in,
at={(0in,0in)},
scale only axis,
xmin=0,
xmax=1,
ymin=0,
ymax=1,
axis line style={draw=none},
ticks=none,
axis x line*=bottom,
axis y line*=left
]
\end{axis}
\end{tikzpicture}%
}\vspace{-0mm}
    \caption{}
    \label{subfig:antenna_2}
\end{subfigure}
\caption{Transmit power versus the number of antennas at the BS, $M=N$, for different schemes.}
\vspace{-5mm}
\label{antenna} 
\end{figure}
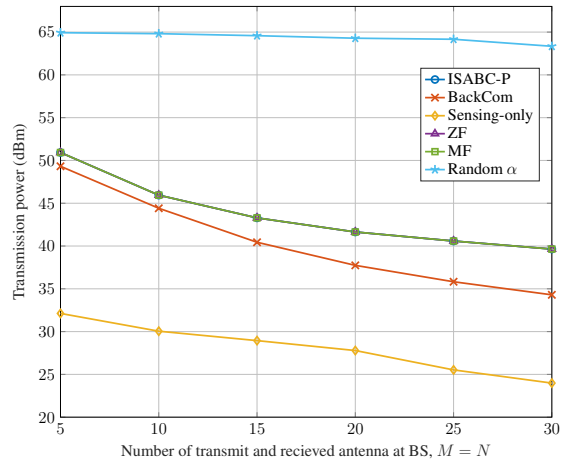
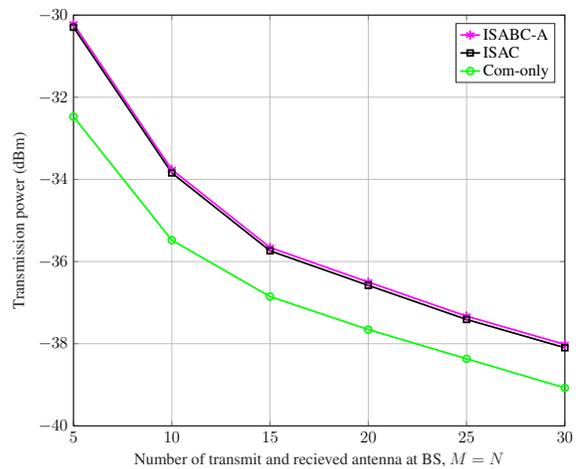

\subsection{Running Time Versus Number of Tags}
Fig. \ref{runn_time} shows the correlation between the execution time and the number of tags, $K$. These data are from  Matlab simulations for an  Intel\textsuperscript{\textregistered} Xeon\textsuperscript{\textregistered} CPU, clocking at \qty{3.5}{\GHz}. Fig. \ref{runn_time} depicts a directly proportional trend between $K$ and the running time for all schemes. This is due to an increase in computation demands as $K$ increases. This tendency emphasizes the complicated challenges of dealing with a larger number of tags, underlining the necessity of efficient algorithms.

Compared to the ISABC-P and sensing-only schemes, the conventional ISAC has a much lower computational time due to the absence of EH constraints, i.e., conventional radar targets do not require EH. The key time-consuming operation in our proposed ISABC-P method is thus satisfying the EH requirements of passive tags. This is evident in the ISABC-A system, which relaxes the EH with battery-powered tags.  For example, ISABC-A only takes \qty{6.76}{\percent} longer than ISAC for algorithm execution with \num{6} tags. Thus, depending on the application scenario, one can select active tags over passive tags at the expense of cost and tag architecture.

\subsection{Transmit Power Versus Number of Antenna at the BS}
We next explore how the  BS transmit power depends on the number of BS transmitter ($M$) and receiver ($N$) antennas, with  $M=N$. As depicted in Fig. \ref{antenna}, a clear and consistent trend emerges across all schemes: increasing $M$ decreases the transmit power. This phenomenon highlights the two benefits of exploiting spatial diversity in  ISABC: increased communication rates and reduced power consumption.
 
Fig. \ref{subfig:antenna_1} and Fig. \ref{subfig:antenna_2} show the BS transmit power requirement as a function of $M$. As per Fig. \ref{subfig:antenna_1}, the random-$\alpha$ benchmark is the most inefficient regarding energy use.  The sensing-only approach requires less transmit power than ISABC-P, as it eliminates the communication performance. Furthermore, sub-optimal MRT- and ZF-based beamformers closely align with our proposed schemes but need CSI. Importantly, ISABC-P allows for sensing with low-cost tags, an essential feature for IoT networks, with only a slight increase in transmit power. For instance, a \qty{3.4}{\percent} increase in transmit power with $M=N=\num{10}$ results in a \qty{75}{\percent} sum rate gain, i.e., user rate + tags' rate + sensing rate.

Fig.~\ref{subfig:antenna_2} reveals that ISAC and purely communication-focused methods require less BS transmit power than our ISABC-P. The latter needs more BS  power to deliver sufficient power at the tags for EH. Of course, this is the cost of using fully passive tags. Nevertheless, utilizing active tags, as in ISABC-A, can obviate the need for this additional power at the BS. For instance, with $M=N=\num{10}$,  ISABC-A only requires a \qty{0.24}{\percent} increase in transmit power to provide a \qty{75}{\percent} sum rate gain over conventional ISAC. On the other hand, active tags with complex tag designs are more expensive than passive tags, and the type of tag used may vary depending on the application.

\begin{figure}[!t]\vspace{-0mm}	
\centering
\fontsize{14}{14}\selectfont 
    \resizebox{.55\totalheight}{!}{
%
%
\definecolor{mycolor1}{rgb}{0.00000,0.44700,0.74100}%
\definecolor{mycolor2}{rgb}{0.85000,0.32500,0.09800}%
\definecolor{mycolor3}{rgb}{0.92900,0.69400,0.12500}%
\definecolor{mycolor4}{rgb}{0.49400,0.18400,0.55600}%
\begin{tikzpicture}

\begin{axis}[%
width=5.877in,
height=4.917in,
at={(0.986in,0.664in)},
scale only axis,
xmin=3,
xmax=15,
xtick={ 3,  6,  9, 12, 15},
xlabel style={font=\color{white!15!black}},
xlabel={Number of tags, $K$},
ymin=-40,
ymax=60,
ylabel style={font=\color{white!15!black}},
ylabel={Transmission power (dBm)},
axis background/.style={fill=white},
xmajorgrids,
ymajorgrids,
legend style={at={(0.975,0.7)}, legend cell align=left, align=left, draw=white!15!black}
]
\addplot [color=mycolor1, line width=1.5pt, mark size=3.0pt, mark=o, mark options={solid, mycolor1}]
  table[row sep=crcr]{%
3	44.8551\\
6	47.0092\\
9	48.4329\\
12	50.9046\\
15	52.3575\\
};
\addlegendentry{ISABC-P}

\addplot [color=mycolor2, line width=1.5pt, mark size=3.5pt, mark=triangle, mark options={solid, mycolor2}]
  table[row sep=crcr]{%
3	-33.2661\\
6	-28.8396\\
9	-26.0763\\
12	-24.2622\\
15	-22.7444\\
};
\addlegendentry{ISABC-A}

\addplot [color=mycolor3, line width=1.5pt, mark size=2.5pt,  mark=square, mark options={solid, mycolor3}]
  table[row sep=crcr]{%
3	-33.7661\\
6	-29.3396\\
9	-26.5763\\
12	-24.7622\\
15	-23.2444\\
};
\addlegendentry{ISAC}

\addplot [color=mycolor4, line width=1.5pt, mark size=3.5pt, mark=diamond, mark options={solid, mycolor4}]
  table[row sep=crcr]{%
3	29.0684\\
6	30.2208\\
9	33.1711\\
12	35.5239\\
15	36.6522\\
};
\addlegendentry{Sensing-only}

\end{axis}

\begin{axis}[%
width=7.583in,
height=6.033in,
at={(0in,0in)},
scale only axis,
xmin=0,
xmax=1,
ymin=0,
ymax=1,
axis line style={draw=none},
ticks=none,
axis x line*=bottom,
axis y line*=left
]
\end{axis}
\end{tikzpicture}%
}\vspace{-0mm}
    \caption{Transmit power versus the number of tags, $K$.}
	\label{num_tag} \vspace{-5mm}
\end{figure}
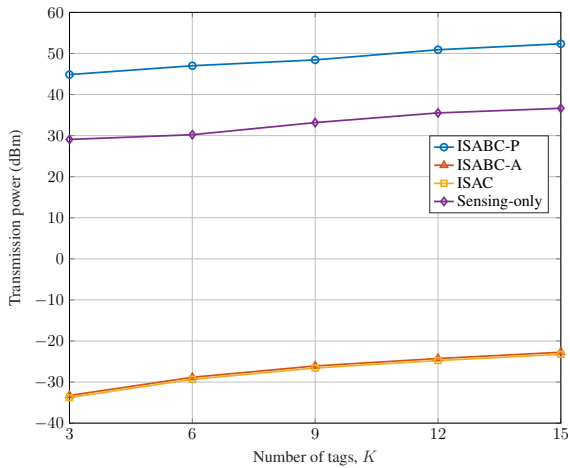

\subsection{Transmit Power Versus Number of Tags}
Fig. \ref{num_tag} delves into the nuanced interplay between the number of tags ($K$) and the BS transmit power requirements. Intuitively, Fig. \ref{num_tag} suggests a linear correlation between these two variables across all schemes. In contrast to conventional ISAC, ISABC-P and sensing-only approaches require the BS to increase power to ensure sufficient EH for the  tags (constraint \eqref{P1:EH}). However, as previously stated, utilizing active tags (ISABC-A) can minimize the increased power demand compared to standard ISAC systems. This comes at a higher tag cost. 

\begin{figure}[!t]\vspace{-0mm}	
\centering
\fontsize{14}{14}\selectfont 
    \resizebox{.55\totalheight}{!}{
%
%
\definecolor{mycolor1}{rgb}{0.00000,0.44700,0.74100}%
\definecolor{mycolor2}{rgb}{0.85000,0.32500,0.09800}%
\definecolor{mycolor3}{rgb}{0.92900,0.69400,0.12500}%
\definecolor{mycolor4}{rgb}{0.49400,0.18400,0.55600}%
\begin{tikzpicture}

\begin{axis}[%
width=5.877in,
height=4.917in,
at={(0.986in,0.664in)},
scale only axis,
xmin=1,
xmax=9,
xtick={1, 3, ..., 9},
xlabel style={font=\color{white!15!black}},
xlabel={Targeted SINR of reader/user, $\Gamma_u^{th}$},
ymin=-40,
ymax=60,
ylabel style={font=\color{white!15!black}},
ylabel={Transmission power (dBm)},
axis background/.style={fill=white},
xmajorgrids,
ymajorgrids,
legend style={at={(0.97,0.5)}, anchor=east, legend cell align=left, align=left, draw=white!15!black}
]
\addplot [color=mycolor1, line width=1.5pt, mark size=3.0pt, mark=o, mark options={solid, mycolor1}]
  table[row sep=crcr]{%
1	48.8901\\
3	49.9113\\
5	51.923\\
7	53.2558\\
9	54.5558\\
};
\addlegendentry{ISABC-P}

\addplot [color=mycolor2, line width=1.5pt, mark size=3.5pt, mark=triangle, mark options={solid, mycolor2}]
  table[row sep=crcr]{%
1	-33.2661\\
3	-26.7457\\
5	-24.15753\\
7	-22.75363\\
9	-21.5613\\
};
\addlegendentry{ISABC-A}

\addplot [color=mycolor3, line width=1.5pt, mark size=2.5pt, mark=square, mark options={solid, mycolor3}]
  table[row sep=crcr]{%
1	-34.5\\
3	-30.7457\\
5	-27.2653\\
7	-26.6363\\
9	-24.8613\\
};
\addlegendentry{Com-only}

\addplot [color=mycolor4, line width=1.5pt, mark size=3.5pt,  mark=diamond, mark options={solid, mycolor4}]
  table[row sep=crcr]{%
1	31.1\\
3	31.1\\
5	31.1\\
7	31.1\\
9	31.1\\
};
\addlegendentry{Sensing-only}

\end{axis}

\begin{axis}[%
width=7.583in,
height=6.033in,
at={(0in,0in)},
scale only axis,
xmin=0,
xmax=1,
ymin=0,
ymax=1,
axis line style={draw=none},
ticks=none,
axis x line*=bottom,
axis y line*=left
]
\end{axis}
\end{tikzpicture}%
}\vspace{-0mm}
    \caption{Transmit power versus the targeted SINR (rate) of the reader/user, $\Gamma_u^{\thr}$. }
	\label{min_rate_fig} \vspace{-5mm}
\end{figure}
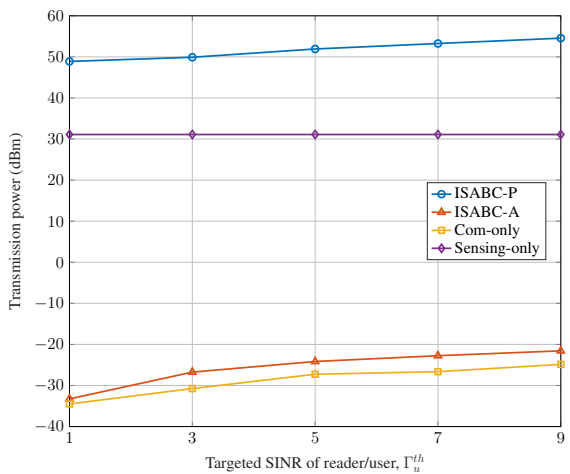

\subsection{Transmit Power Versus User SINR  Requirement} 
In Fig. \ref{min_rate_fig}, we examine the relationship between the BS transmit power and the targeted user SINR, $\Gamma_u^{\thr}$. 
This captures the impact of variations in the targeted user rate on the BS transmit power.
Benchmarks for communication-only, sensing-only, and  ISABC-P/A schemes are also plotted to offer a comprehensive comparison.

The BS requires minimal power in the communication-only benchmark to meet the user's rate demand. The BS maintains consistent power for EH and sensing rates without communication functionality in the sensing-only approach with passive tags. In the proposed ISABC-P, the BS requires higher transmit power, supporting primary communication and sensing services vital for future ambient-powered IoT networks. Alternatively, active tags (ISABC-A) can eliminate the higher power requirement at the BS, although active tags are costlier and more complex to maintain than passive tags \cite{Diluka2022, Rezaei2023Coding}.

\subsection{Communication Impairments}
{The performance and reliability of wireless links can be hampered by impairments, particularly CSI and SI cancellation errors. CSI discrepancies affect processes like signal reception and beamforming, arising from inaccuracies in channel state estimation. On the other hand, SI cancellation errors occur when SI is not eliminated, hindering incoming signal reception. Both errors critically impair communication dynamics \mbox{\cite{Shengli2004}}.  }

{Due to the significance of these errors, we explore their detailed ramifications on Algorithm \mbox{\ref{alg:AO:maxmin}}. Fig. \mbox{\ref{csi_imp}} examines the consequences of imperfect CSI and SI cancellation on  Algorithm \mbox{\ref{alg:AO:maxmin}}.  The relation $\hat{x} = x + e$ models the channel estimation process. Here, the true channel is denoted by $x\in \{f_m, v_k\}$, $\forall m\in \{1,\ldots, M\}, \forall k \in \mathcal{K}$, and the term $e$ characterizes the Gaussian-distributed estimation error with zero mean, mathematically expressed as $ e \sim \mathcal{N} (0, \sigma_e^2)$. A pivotal parameter in this context is the error variance, which adheres to the inequality $\sigma_e^2  \triangleq \eta |x|^2$. Here, the coefficient $\eta$ serves as a metric to gauge the magnitude of CSI error. Fig. \mbox{\ref{csi_imp}} displays the correlation between the transmit power and $\eta$. As one increases, so does the other, requiring more transmit power across all communication schemes.}

{Shifting our attention to the nuances of imperfect SI cancellation, it is essential to understand and counteract its effects. We replace the residual SI term, given by $\lambda \vert \mathbf{f}^{\rm{H}} \mathbf{w}\vert^2$, into the SINR of $T_k$. The variable $\lambda\in[0, 1]$ indicates the degree of imperfection in the SI cancellation process. Further experiments shed light on the transmission power in scenarios with varied values of residual SI, especially under the influence of distinct CSI errors. The transmit power is sacrificed as $\lambda$ increases.}

\begin{figure}[!t]\vspace{-0mm}	
\centering
\fontsize{14}{14}\selectfont 
    \resizebox{.55\totalheight}{!}{
%
%
\definecolor{mycolor1}{rgb}{0.00000,0.44700,0.74100}%
\definecolor{mycolor2}{rgb}{0.85000,0.32500,0.09800}%
\definecolor{mycolor3}{rgb}{0.92900,0.69400,0.12500}%
\definecolor{mycolor4}{rgb}{0.49400,0.18400,0.55600}%
\begin{tikzpicture}

\begin{axis}[%
width=5.877in,
height=4.917in,
at={(0.986in,0.664in)},
scale only axis,
xmin=0,
xmax=1,
xtick={0, 0.2, ..., 1},
xlabel style={font=\color{white!15!black}},
xlabel={CSI imperfection, $\eta$},
ymin=48,
ymax=48.8,
ylabel style={font=\color{white!15!black}},
ylabel={Transmission power (dBm)},
axis background/.style={fill=white},
xmajorgrids,
ymajorgrids,
legend style={at={(0.03,0.97)}, anchor=north west, legend cell align=left, align=left, draw=white!15!black}
]
\addplot [color=mycolor1, line width=1.5pt, mark size=3.0pt, mark=o, mark options={solid, mycolor1}]
  table[row sep=crcr]{%
0	48.31\\
0.2	48.42\\
0.4	48.5166\\
0.6	48.6093\\
0.8	48.6961\\
1	48.7827\\
};
\addlegendentry{ISABC-P, $\lambda$ = -80 dB}

\addplot [color=mycolor2, line width=1.5pt, mark size=2.5pt,  mark=square, mark options={solid, mycolor2}]
  table[row sep=crcr]{%
0	48.269\\
0.2	48.38\\
0.4	48.4794\\
0.6	48.5685\\
0.8	48.66204\\
1	48.7556\\
};
\addlegendentry{ISABC-P, $\lambda$ = -90 dB}

\addplot [color=mycolor3, line width=1.5pt, mark size=3.5pt, mark=diamond, mark options={solid, mycolor3}]
  table[row sep=crcr]{%
0	48.2268\\
0.2	48.3301\\
0.4	48.424\\
0.6	48.52\\
0.8	48.61\\
1	48.7\\
};
\addlegendentry{ISABC-P, $\lambda$ = -110 dB}

\addplot [color=mycolor4, line width=1.5pt, mark size=3.5pt, mark=triangle, mark options={solid, mycolor4}]
  table[row sep=crcr]{%
0	48.0916\\
0.2	48.1921\\
0.4	48.297\\
0.6	48.400912\\
0.8	48.4994\\
1	48.5948\\
};
\addlegendentry{ISABC-P, perfect SI cancellation}

\end{axis}

\begin{axis}[%
width=7.583in,
height=6.033in,
at={(0in,0in)},
scale only axis,
xmin=0,
xmax=1,
ymin=0,
ymax=1,
axis line style={draw=none},
ticks=none,
axis x line*=bottom,
axis y line*=left
]
\end{axis}
\end{tikzpicture}%
}\vspace{-0mm}
    \caption{{Transmit power versus CSI imperfection $\eta$, for various residual SIC values. }}
	\label{csi_imp} \vspace{-5mm}
\end{figure}
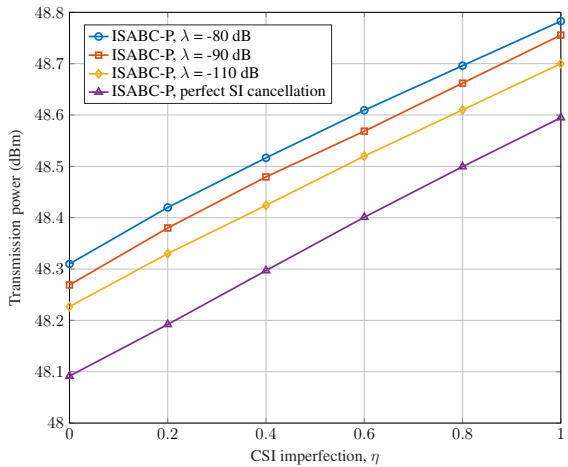

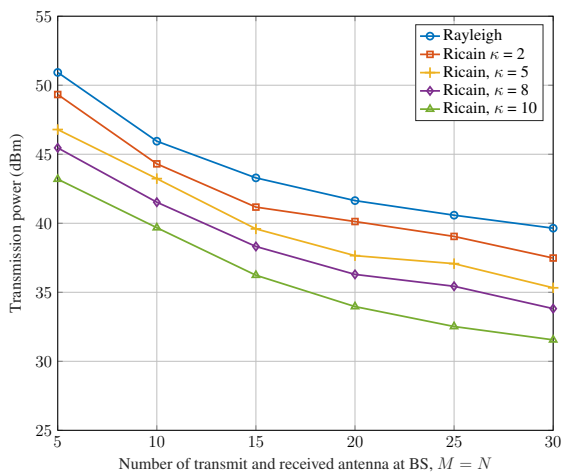
\begin{figure}[!t]\vspace{-0mm}	
\centering
\fontsize{14}{14}\selectfont 
    \resizebox{.55\totalheight}{!}{
%
%
\definecolor{mycolor1}{rgb}{0.00000,0.44700,0.74100}%
\definecolor{mycolor2}{rgb}{0.85000,0.32500,0.09800}%
\definecolor{mycolor3}{rgb}{0.92900,0.69400,0.12500}%
\definecolor{mycolor4}{rgb}{0.49400,0.18400,0.55600}%
\definecolor{mycolor5}{rgb}{0.46600,0.67400,0.18800}%
\begin{tikzpicture}

\begin{axis}[%
width=5.877in,
height=4.916in,
at={(0.972in,0.66in)},
scale only axis,
xmin=5,
xmax=30,
xtick={5, 10, ..., 30},
xlabel style={font=\color{white!15!black}},
xlabel={Number of transmit and received antenna at BS, $M = N$},
ymin=25,
ymax=55,
ytick={25, 30, ..., 55},
ylabel style={font=\color{white!15!black}},
ylabel={Transmission power (dBm)},
axis background/.style={fill=white},
xmajorgrids,
ymajorgrids,
legend style={legend cell align=left, align=left, draw=white!15!black}
]
\addplot [color=mycolor1, line width=1.5pt, mark size=3.0pt, mark=o, mark options={solid, mycolor1}]
  table[row sep=crcr]{%
5	50.9271\\
10	45.9448\\
15	43.2884\\
20	41.641\\
25	40.585\\
30	39.6431\\
};
\addlegendentry{Rayleigh}

\addplot [color=mycolor2, line width=1.5pt, mark size=2.5pt,  mark=square, mark options={solid, mycolor2}]
  table[row sep=crcr]{%
5	49.3279\\
10	44.305\\
15	41.1705\\
20	40.1253\\
25	39.047\\
30	37.481\\
};
\addlegendentry{$\text{Ricain }\kappa\text{ = 2}$}

\addplot [color=mycolor3, line width=1.5pt, mark size=4.5pt, mark=+, mark options={solid, mycolor3}]
  table[row sep=crcr]{%
5	46.79052\\
10	43.2328\\
15	39.5933\\
20	37.6533\\
25	37.0741\\
30	35.3246\\
};
\addlegendentry{$\text{Ricain, }\kappa\text{ = 5}$}

\addplot [color=mycolor4, line width=1.5pt, mark size=3.5pt, mark=diamond, mark options={solid, mycolor4}]
  table[row sep=crcr]{%
5	45.4624\\
10	41.5257\\
15	38.3178\\
20	36.2953\\
25	35.4356\\
30	33.8162\\
};
\addlegendentry{$\text{Ricain, }\kappa\text{ = 8}$}

\addplot [color=mycolor5, line width=1.5pt, mark size=3.5pt, mark=triangle, mark options={solid, mycolor5}]
  table[row sep=crcr]{%
5	43.1941\\
10	39.6785\\
15	36.2354\\
20	33.9641\\
25	32.5201\\
30	31.5536\\
};
\addlegendentry{$\text{Ricain, }\kappa\text{ = 10}$}

\end{axis}

\begin{axis}[%
width=7.479in,
height=6in,
at={(0in,0in)},
scale only axis,
xmin=0,
xmax=1,
ymin=0,
ymax=1,
axis line style={draw=none},
ticks=none,
axis x line*=bottom,
axis y line*=left
]
\end{axis}
\end{tikzpicture}%
}\vspace{-0mm}
    \caption{{Transmit power for different fading channels.}}
	\label{Rician_fig} \vspace{-5mm}
\end{figure}

{In Fig. {\ref{Rician_fig}}, the performance evaluation under varying channel models is depicted for our proposed scheme. In particular, we evaluate the transmit power performance for Rayleigh and Rician fading channels with different Rician factors ($\kappa$). The Rician fading channel model is represented as}
\begin{subequations}
\begin{align}
    \qf  & = \sqrt{\frac{\kappa}{\kappa+1}} \qf^{\text{LoS}} + \sqrt{\frac{1}{\kappa+1}} \qf^{\text{NLoS}},  \\
    v_k & = \sqrt{\frac{\kappa}{\kappa+1}} v_k^{\text{LoS}} + \sqrt{\frac{1}{\kappa+1}} v_k^{\text{NLoS}},
\end{align}
\end{subequations}
{where $v_k^{\text{LoS}}=1$ and $\qf^{\text{LoS}}$ which is based on (2) are the deterministic LoS components that correspond to the direct path between the transmitter and receiver. Also, $\qf^{\text{NLoS}}$ and $v_k^{\text{NLoS}}$ are the non-LoS (NLoS) components that follow the Rayleigh fading model. As shown in Fig. {\ref{Rician_fig}}, high $\kappa$ factors lead to lower transmission power than Rayleigh channels. For example, channels with Rician factors ${\kappa=\num{2}}$ and ${\kappa=\num{5}}$ result in a {\qty{3.64}{\percent}} and {\qty{3.64}{\percent}} reduction in transmit power with ${M=N=\num{20}}$, respectively, when compared to Rayleigh channels. This is due to the LoS component in the propagation of the signal, which contributes to more efficient power utilization.}

\section{Conclusion}In this study, we introduced the concept of ISABC, an innovative system that combines \bc and sensing with an FD BS. The FD BS plays a dual role, serving as a sensor through backscatter tag detection and enabling user communication. Using the AO technique, we derived precise communication (both primary and backscatter) and sensing rates, focusing on minimizing power consumption at the BS. This approach empowers passive and active tags for sensing and communication, opening doors to new possibilities. Future research can explore advanced modulation and machine learning to enhance ISABC's adaptability in urban environments, aligning with IoT and 6G technologies for comprehensive networks. Addressing multi-user scenarios, interference mitigation, and practical hardware implementation are crucial for real-world advancements.

\appendices

\section{Proof of Theorem 1}\label{Theorem_1}
Recall that we divide (P1)   into three sub-problems to  optimize  $\left(\qalpha := \{\alpha_k\}_{k\in \mathcal{K}}\right)$, received beamformers $\left( \qU :=\{\qu_k\}_{k\in \mathcal{K}} \right)$, and transmit beamformers $\left(\qP= \{\qw, \qS\}\right)$, via solving problems \eqref{opt:uk}, \eqref{P5}, and \eqref{P7}, while keeping the other two blocks of variables fixed. Let us define $F(\qU, \qalpha, \qP)$ as a function of $\qU$, $\qalpha$, and $\qP$ for the objective value of \eqref{P2}. First, in step $3$ of Algorithm \ref{alg:AO:maxmin} with fixed variables $\qalpha^{(i)}$ and $\qP^{(i)}$, $\qU^{(i+1)}$ is the optimal solution that minimizes the value of the objective function. Accordingly, we have
\begin{equation}\label{43}
F(\qalpha^{(i)},\qU^{(i+1)},\qP^{(i)}) \leq F(\qalpha^{(i)},\qU^{(i)},\qP^{(i)}).
\end{equation}
Next, in step $4$ of Algorithm \ref{alg:AO:maxmin}, $\qP^{(i+1)}$ is the optimal transmit beamformers with given variables $\qalpha^{(i)}$ and $\qU^{(i+1)}$ to minimize $F$ via solving \eqref{P5}. Thus, it guarantees that
\begin{equation}
F(\qalpha^{(i)},\qU^{(i+1)},\qP^{(i+1)}) \leq F(\qalpha^{(i)},\qU^{(i+1)},\qP^{(i)}).
\end{equation}
Finally, in step $5$ of Algorithm \ref{alg:AO:maxmin} with the given $\qP^{(i+1)}$ and  $\qU^{(i+1)}$, problem \eqref{P7} is solved to obtain an optimal solution for $\qalpha^{(i)}$, which yields:
\begin{equation}\label{45}
F(\qalpha^{(i+1)},\qU^{(i+1)},\qP^{(i+1)}) \leq F(\qalpha^{(i)},\qU^{(i+1)},\qP^{(i+1)}).
\end{equation}
According to \eqref{43}--\eqref{45}, we can conclude that \cite{ZARGARI2021101413,Shayan_Zargari}:
\begin{equation} 
F(\qalpha^{(i+1)},\qU^{(i+1)},\qP^{(i+1)}) \leq F(\qalpha^{(i)},\qU^{(i)},\qP^{(i)}).
\end{equation}
For problem \eqref{P2}, the objective values of Algorithm \ref{alg:AO:maxmin} monotonically decrease with each iteration, always remaining non-negative. This consistency, combined with the design choice where each iteration starts from the previous one's end, ensures the convergence of the algorithm. Essentially, the objective function will either decrease or stay the same until it meets the convergence criteria, resulting in a stable solution. Thus, the proof is completed.

\bibliographystyle{IEEEtran}
\bibliography{IEEEabrv,ref}

\end{document}